\documentclass[a4paper, 12pt]{amsart}

\usepackage{amsmath}
\usepackage{amsthm}
\usepackage{amssymb}
\usepackage{mathtools}
\usepackage{color}
\usepackage{graphicx}
\usepackage{siunitx}

\newcommand{\R}{\mathbb{R}}
\newcommand{\N}{\mathbb{N}}
\newcommand{\Z}{\mathbb{Z}}

\newcommand{\thetalin}{\hat\theta^{\mathrm{lin},N}_0}
\newcommand{\thetapol}{\hat\theta^{\mathrm{pol},N}_0}
\newcommand{\thetafull}{\hat\theta^{\mathrm{full},N}_0}
\newcommand{\thetatwo}{\hat\theta^{2,N}_0}
\newcommand{\thetathree}{\hat\theta^{3,N}_0}
\newcommand{\thetafour}{\hat\theta^{4,N}_0}

\newtheorem{theorem}{Theorem}
\newtheorem{lemma}[theorem]{Lemma}
\newtheorem{proposition}[theorem]{Proposition}

\begin{document}

\title[Diffusivity estimation for activator-inhibitor 
models]{Diffusivity Estimation for Activator-Inhibitor Models: Theory and Application to Intracellular Dynamics of the Actin Cytoskeleton}
\author[G. Pasemann, S. Flemming, S. Alonso, C. Beta, 
W. Stannat]{Gregor Pasemann\footnotesize{$^1$}, 
Sven Flemming\footnotesize{$^{2}$}, 
Sergio Alonso\footnotesize{$^{3}$}, \\ 
Carsten Beta\footnotesize{$^{2}$}, 
Wilhelm Stannat\footnotesize{$^{1}$}}

\email{pasemann@math.tu-berlin.de \\
svenflem@uni-potsdam.de \\
s.alonso@upc.edu \\
beta@uni-potsdam.de \\
stannat@math.tu-berlin.de}
\date{\today}

\begin{abstract}
A theory for diffusivity estimation for spatially extended activator-inhibitor dynamics modelling the evolution of intracellular signaling networks is developed in the mathematical framework of stochastic reaction-diffusion systems. In order to account for model uncertainties, we extend the results for parameter estimation for semilinear stochastic partial differential equations, as developed in \cite{PasemannStannat20}, to the problem of joint estimation of diffusivity and parametrized reaction terms. 
Our theoretical findings are applied to the estimation of effective diffusivity of signaling components contributing to intracellular dynamics of the actin cytoskeleton in the model organism {\it Dictyostelium discoideum}.
\end{abstract}



\maketitle

\footnotetext[1]{Institute of Mathematics, Technische Universit\"at 
Berlin, 10623 Berlin, Germany}
\footnotetext[2]{Institute of Physics and Astronomy, University of Potsdam, 14476 Potsdam, Germany} 
\footnotetext[3]{Department of Physics, Universitat Politècnica de Catalunya, 08034 Barcelona, Spain}

\section{Introduction}

The purpose of this paper is to develop the mathematical 
theory for statistical inference methods for the parameter 
estimation of stochastic reaction-diffusion systems modelling 
spatially extended signaling networks 
in cellular systems.
Such signaling networks are one of the central topics in cell biology and biophysics as they provide the basis for essential processes including cell division, cell differentiation, and cell motility~\cite{devreotes_excitable_2017}.
Nonlinearities in these network may cause rich spatiotemporal behavior including the emergence of oscillations and waves~\cite{beta_intracellular_2017}.
Furthermore, alterations and deficiencies in the network topology can explain many pathologies and play a key role in diseases such as cancer~\cite{condeelis_great_2005}.
Here we present a method to estimate both diffusivity and reaction terms of a stochastic reaction-diffusion system, given space-time structured data of local concentrations of signaling components. 
We mainly focus on the estimation of diffusivity, whose precision can be increased by simultaneous calibration of the reaction terms. 
To test this approach, we use fluorescence microscopy recordings of the actin dynamics in the cortex of cells of the social amoeba {\it Dictyostelium discoideum}, 
a well-established model organism for the study of a wide range of actin-dependent processes~\cite{annesley_dictyostelium_2009}.
A recently introduced stochastic reaction-diffusion model could reproduce many features of the dynamical patterns observed in the cortex of these cells including excitable and bistable states~\cite{alonso_modeling_2018,FlemmingFontAlonsoBeta20,moreno_modeling_2020}.
In combination with the experimental data, this model will serve as a specific test case to exemplify our mathematical approach.
Since in real-world applications the available data will not 
allow for calibrating and validating detailed mathematical 
models, in this paper we will be primarily interested in 
minimal models that are still capable of generating all 
observed dynamical features at correct physical magnitudes. 
The developed estimation techniques should in practice 
be as robust as possible w.r.t. uncertainty and even 
misspecification of the unknown real dynamics. 

The impact of diffusion and reaction 
in a given model 
will be of fundamentally different structure and it is one of the 
main mathematical challenges to separate these impacts 
in the data 
in order to come to valid parameter estimations. 
On the more mathematical side, diffusion corresponds to a 
second order partial differential operator --- resulting in a 
strong spatial coupling in the given data, whereas the 
reaction corresponds to a lower order, in fact $0$-order, in 
general resulting in highly nonlinear local interactions in the data. 
For introductory purposes, let us assume that our data is 
given in terms of a space and time-continuous field $X(t,x)$ 
on $[0,T]\times\mathcal{D}$, where $T$ is the terminal time 
of our observations and $\mathcal{D}\subset\R^2$ a 
rectangular domain that corresponds to a chosen data-segment 
in a given experiment. 
Although in practice the given data will be discrete w.r.t. both space 
and time, we will be interested in applications where the resolution is high enough in order to approximate the data by such a continuous field. 
Our standing assumption is that $X(t, x)$ is generated by a dynamical system of the form
\begin{align}
\label{eq:IntroGenericEquation}
	\partial_t X(t, x) = \theta_0\Delta X(t, x) + 
	\mathcal{F}_X(t, x),
\end{align}
where $\Delta$ is the Laplacian, given by $\Delta X(t, x) = \partial_{x_1}^2X(t, x) + \partial_{x_2}^2X(t, x)$, $x=(x_1,x_2)$, which captures the diffusive spreading in the dynamics of $X(t, x)$. The intensity of the diffusion is given by the diffusivity $\theta_0$. Finally, 
$\mathcal{F}$ is a generic term, depending on the solution field $X(t, x)$, which describes 
all non-diffusive effects present in $X(t, x)$, whether they are known or unknown. 
A natural approach to extract $\theta_0$ from the data is to use a ``cutting-out estimator'' 
of the form
\begin{align}\label{eq:IntroGenericEstimator}
	\hat\theta_0 = \frac{\int_0^T\int_{\mathcal{D}} 
	Y(t, x)\partial_t X(t, x) 
	\mathrm{d}x\mathrm{d}t}{\int_0^T\int_{\mathcal{D}}Y(t, x) 
	\Delta X(t, x)\mathrm{d}x\mathrm{d}t}
	 = \frac{\int_0^T\langle Y,\partial_tX\rangle\mathrm{d}t}{\int_0^T\langle Y,\Delta X\rangle\mathrm{d}t},
\end{align}
where $Y(t, x)$ is a suitable test function. In the second fraction of \eqref{eq:IntroGenericEstimator}, we use the functional form for readability. In particular, we write $X=X(t, x)$ for the solution field. We will also write $X_t = X(t, \cdot)$ for the (spatially varying) solution field at a fixed time $t$. In order to ease notation, we will use this functional form from now on throughout the paper. It is possible to derive \eqref{eq:IntroGenericEstimator} 
from a least squares approach by minimizing $\theta\mapsto\lVert\partial_tX-\theta\Delta X\rVert^2$ with a suitably chosen norm. If the non-diffusive 
effects described by $\mathcal{F}$ are negligible, we see by 
plugging \eqref{eq:IntroGenericEquation} into 
\eqref{eq:IntroGenericEstimator} that $\hat\theta_0$ is 
close to $\theta_0$. If a sound approximation 
$\overline{\mathcal{F}}$ to $\mathcal{F}$ is known, the 
estimator can be made more precise by substituting 
$\partial_tX$ by $\partial_tX-\overline{\mathcal{F}}_X$ in 
\eqref{eq:IntroGenericEstimator}. A usual choice for $Y$ is a 
reweighted spectral cutoff of $X$. 

Under additional model assumptions, e.g. if 
\eqref{eq:IntroGenericEquation} is in fact a stochastic 
partial differential equation (SPDE) driven by Gaussian white noise, a rather developed parameter estimation theory for 
$\theta_0$ has been established in \cite{PasemannStannat20} on the basis of maximum likelihood estimation (MLE). 

In this paper, we are interested in further 
taking into account also those parts of $\mathcal{F}_X$ 
corresponding to local nonlinear reactions. As a particular example, we will focus on a recently introduced stochastic reaction-diffusion system of FitzHugh--Nagumo type that captures many aspects of the dynamical wave patterns observed in the cortex of motile amoeboid cells~\cite{FlemmingFontAlonsoBeta20},
\begin{align}
	\partial_t U &= D_U\Delta U + k_1 U(u_0-U)(U-u_0a)- k_2V+\xi, \label{eq:IntroActivator}\\
	\partial_t V &= D_V\Delta V  
	+ \epsilon(bU - V).\label{eq:IntroInhibitor}
\end{align}
Here, we identify $\theta_0=D_U$ and the only observed data 
is the activator variable, i.e. $X=U$. 

Therefore, in this example the non-diffusive part of the dynamics will be further decomposed as  
$\mathcal{F} = F + \xi$, where $\xi$ is Gaussian white noise 
and $F=F(U)$ encodes the non-Markovian reaction dynamics of 
the activator. The inhibitor component $V$ in the above 
reaction-diffusion system is then incorporated for minimal 
modelling purposes to allow the formation of travelling 
waves in the activator variable $U$ that are indeed 
observed in the time evolution of the actin concentration. 
This model and its dynamical features is explained in detail in Section \ref{sec:ExampleFitzHughNagumo}.

As noted before, it is desirable to include this additional 
knowledge into the estimation procedure 
\eqref{eq:IntroGenericEstimator} by subtracting a suitable 
approximation $\overline{\mathcal{F}}$ of the --- in practice ---
unknown $\mathcal{F}$. Although \eqref{eq:IntroActivator}, 
\eqref{eq:IntroInhibitor} suggest an explicit parametric  
form for $\overline{\mathcal{F}}$, it is a priori not clear 
how to quantify the nuisance parameters appearing in the 
system. Thus an (approximate) model for the data is known 
{\it qualitatively}, based on the observed dynamics, but not 
{\it quantitatively}. In order to resolve this issue, we 
extend \eqref{eq:IntroGenericEstimator} and adopt a joint 
maximum likelihood estimation of $\theta_0$ and various 
nuisance parameters. 

The field of statistical inference for SPDEs is rapidly growing, see \cite{Cialenco18} for a recent survey. 
The spectral approach to drift estimation was pioneered by \cite{HubnerKhasminskiiRozovskii93,HuebnerRozovskii95} and subsequently extended by various works, see e.g. \cite{HuebnerLototskyRozovskii97, LototskyRozovskii99, LototskyRozovskii00} for the case of non-diagonalizable linear evolution equations. In \cite{CialencoGlattHoltz11}, the stochastic Navier-Stokes equations have been analyzed as a first example of a nonlinear evolution equation. This has been generalized by \cite{PasemannStannat20} to semilinear SPDEs. 
Joint parameter estimation for linear evolution equations is treated in \cite{Huebner93,Lototsky03}. 
Besides the spectral approach, other measurement schemes have been studied. 
See e.g. \cite{PospisilTribe07, BibingerTrabs17, BibingerTrabs19, Chong18, Chong19, KhalilTudor19, CialencoHuang19, CialencoDelgadoVencesKim19, CialencoKim20, KainoUchida19} for the case of discrete observations in space and time. Recently, the so-called local approach has been worked out in \cite{AltmeyerReiss19} for linear equations, 
was subsequently generalized in \cite{AltmeyerCialencoPasemann20} to the semilinear case 
and applied to a stochastic cell repolarization model in \cite{AltmeyerBretschneiderJanakReiss20}. 

The paper is structured as follows: In Section \ref{sec:Theory} we give a theory for joint diffusivity and reaction parameter estimation for a class of semilinear SPDEs and study the spatial high-frequency asymptotics. In Section \ref{sec:data}, the biophysical context for these models is discussed. The performance of our method on simulated and real data is evaluated in Section \ref{sec:data_analysis}.

\section{Maximum Likelihood Estimation for \\ Activator-Inhibitor Models}\label{sec:Theory}

In this section we develop a theory for parameter estimation for a class of semilinear SPDE using a maximum likelihood ansatz. The application we are aiming at is an activator-inhibitor model as in \cite{FlemmingFontAlonsoBeta20}. More precisely, we show under mild conditions that the diffusivity of such a system can be identified in finite time given high spatial resolution and observing only the activator component. 

\subsection{The Model and Basic Properties}

Let us first introduce the abstract mathematical setting in 
which we are going to derive our main theoretical results. 
We work in spatial dimension $d\geq 1$. 
Given a bounded domain $\mathcal{D}=[0,L_1] \times 
\dots \times [0,L_d] \subset \R^d$, $L_1,\dots,L_d>0$, we 
consider the following parameter estimation problem for the 
semilinear SPDE
\begin{align}\label{eq:SPDEmodelGeneral}
	\mathrm{d}X_t = \left(\theta_0\Delta X_t 
	+ F_{\theta_1,\dots,\theta_K}(X)\right)\mathrm{d}t 
	+ B\mathrm{d}W_t
\end{align}
with periodic boundary conditions for $\Delta$ on the 
Hilbert space $H=\bar L^2(\mathcal{D}) = \{u \in 
L^2 (\mathcal{D}) |\int_\mathcal{D} u\mathrm{d}x=0\}$, 
together with initial condition $X_0\in H$. 
We allow the nonlinear term $F$ to depend on additional (nuisance) parameters $\theta_1,\dots,\theta_K$ and write $\theta=(\theta_0,\dots,\theta_K)^T$, $\theta_{1:K}=(\theta_1,\dots,\theta_K)$ for short. 
Without further mentioning it, we assume that 
$\theta\in\Theta$ for a fixed 
parameter space 
$\Theta$, e.g. $\Theta=\R_+^K$. 
Next, $W$ is a cylindrical Wiener process modelling Gaussian space-time white noise, that is, $\mathbb{E}[\dot W(t, x)]=0$ and $\mathbb{E}[\dot W(t, x)\dot W(s, y)]=\delta(t-s)\delta(x-y)$. In order to introduce spatial correlation, we use a dispersion operator of the form $B=\sigma(-\Delta)^{-\gamma}$ with $\sigma>0$ and $\gamma>d/4$. 
Here, $\sigma$ is the noise intensity, and $\gamma$ quantifies the decay of the noise for large frequencies in Fourier space. In addition, $\gamma$ determines the spatial smoothness of $X$, see Section \ref{sec:BasicRegularityResults}. The condition $\gamma>d/4$ ensures that the covariance operator $BB^T$ is of trace class, which is a standard assumption for well-posedness of \eqref{eq:SPDEmodelGeneral}, cf. \cite{LiuRockner15}. 
Denote by $(\lambda_k)_{k\geq 0}$ the 
eigenvalues of $-\Delta$, ordered increasingly, with 
corresponding eigenfunctions $(\Phi_k)_{k\geq 0}$. It is 
well-known \cite{Weyl11, Shubin01} that 
$\lambda_k\asymp\Lambda k^{2/d}$ for a constant $\Lambda>0$, i.e. $\lim_{k\rightarrow\infty}\lambda_k/(\Lambda k^{2/d})=1$. 
The proportionality constant $\Lambda$ is known explicitly (see e.g. \cite[Proposition 13.1]{Shubin01}) and depends on the domain $\mathcal{D}$.
Let $P_N:H\rightarrow H$ be the projection onto the span of the 
first $N$ eigenfunctions, and set $X^N:=P_NX$. 
For later use, we denote by $I$ the identity operator acting on $H$.
For $s\in\R$, we 
write $H^s:=D((-\Delta)^{s/2})$ for the domain of $(-\Delta)^{s/2}$, which is given by
\begin{align*}
    (-\Delta)^{s/2}x=\sum_{k=1}^\infty \lambda_k^{s/2}\langle \Phi_k, x\rangle\Phi_k,
\end{align*}
and abbreviate $|\cdot|_s:=|\cdot|_{H^s}$ for the norm on that space whenever convenient. 
We assume that the initial condition $X_0$ is regular enough, i.e. it satisfies $\mathbb{E}[|X_0|_{s}^p]<\infty$ for any $s\geq 0$, $p\geq 1$, without further mentioning it in the forthcoming statements. We will use the following general class of conditions with $s\geq 0$ in order to describe the regularity of $X$:
\begin{itemize}
	\item[$(A_s)$] For any $p\geq 1$, it holds
		\begin{align}
			\mathbb{E}\left[\sup_{0\leq t\leq T}|X_t|_{s}^p\right] < \infty.
		\end{align}
\end{itemize}
Our standing assumption is that there is a unique solution $X$ to \eqref{eq:SPDEmodelGeneral} such that $(A_0)$ holds. This is a general statement and can be derived e.g. under the assumptions from \cite[Theorem 3.1.5]{LiuRockner15}. 

\subsubsection{An Activator-Inhibitor Model}\label{sec:ExampleFitzHughNagumo}

An important example for our analysis is given by the following FitzHugh--Nagumo type system of equations in $d\leq 2$ (cf. \cite{FlemmingFontAlonsoBeta20}): 
\begin{align}
	\mathrm{d}U_t &= (D_U\Delta U_t + k_1f(\lvert U_t\rvert_{L^2}, U_t) - k_2V_t)\mathrm{d}t + B\mathrm{d}W_t, \label{eq:Activator}\\
	\mathrm{d}V_t &= (D_V\Delta V_t + \epsilon(bU_t - V_t))\mathrm{d}t,\label{eq:Inhibitor}
\end{align}
together with sufficiently smooth initial conditions. 
Here, 
$f$ is a bistable third-order polynomial  
$f(x, u) = u(u_0-u)(u - a(x)u_0)$, and $a\in C^1_b (\R,\R)$ is a bounded and continuously differentiable function with bounded 
derivative. The boundedness condition for $a$ is not  
essential to the dynamics of $U$ and can be realized in 
practice by a suitable cutoff function. \\

The FitzHugh--Nagumo system \cite{Fitzhugh61, Nagumo62} originated as a minimal model capable of generating excitable pulses mimicking action potentials in neuroscience. 
Its two components $U$ and $V$ are called {\it activator} and {\it inhibitor}, resp.

The spatial extension of the Fitzhugh--Nagumo system, obtained via diffusive 
coupling, is used to model propagation of excitable pulses and two-phase dynamics. 
In the case of the two phases, low and high concentration of the activator $U$ are realized as the stable fixed points of the third order polynomial $f$ at $0$ and $u_0$. 
The unstable fixed point  
$a u_0$, $a\in (0,1)$, separates the domains of attraction of the two stable fixed 
points. The interplay between spatial diffusion, causing a smoothing out of 
concentration gradients with rate $D_U$, and the local reaction forcing $f$, 
causing convergence of the activator to one of the stable phases with rate 
$k_1$, leads to the formation of transition phases between regions with low or high 
concentration of $U$. The parameters determine shape and velocity of the transition 
phases, e.g. low values of $a$ enhance the growth of regions with high activator 
concentration. This corresponds to the excitable regime, as explained in 
\cite{FlemmingFontAlonsoBeta20}.

Conversely, a high concentration of the inhibitor $V$ leads to a decay in the 
activator $U$, with rate $k_2$. In the excitable regime, this mechanism leads to 
moving activator wave fronts. The inhibitor is generated with rate $\epsilon b$ in 
the presence of $U$ and decays at rate $\epsilon$. Finally, its spatial evolution 
is determined by diffusion with rate $D_V$. 

Finally, choosing $a$ as a functional depending on the total activator concentration 
introduces a feedback control that allows to stabilize the dynamics. 

A detailed discussion of the relevance for cell biology is given in Section 
\ref{sec:data}. More information on the FitzHugh--Nagumo model and related models can be found in \cite{ErmentroutTerman2010}. 
\\

For this model, we can 
find a representation of the above type \eqref{eq:SPDEmodelGeneral} as follows: Using the 
variation of constants formula, the solution $V$ to 
\eqref{eq:Inhibitor} with initial condition $V_0=0$ can be 
written as $V_t=\epsilon b\int_0^te^{(t-r)(D_V\Delta 
-\epsilon I)}U_r\mathrm{d}r$. Inserting this representation 
into \eqref{eq:Activator} yields the following reformulation 
\begin{equation} 
\label{eq:SPDEmodelLinearFN}
\begin{aligned} 
	\mathrm{d}U_t & = \Big( D_U \Delta U_t + k_1 U_t 
	(u_0 - U_t )(U_t - a u_0) \\ 
	& \quad - k_2 \epsilon b\int_0^t e^{(t-r)( D_V\Delta -\epsilon I)}U_r\mathrm{d}r\Big)\mathrm{d}t 
	+ B\mathrm{d}W_t \\ 
	& = \left(\theta_0\Delta U_t + \theta_1F_1(U_t) + \theta_2F_2(U_t) + \theta_3F_3(U)(t) \right)\mathrm{d}t + B\mathrm{d}W_t 
\end{aligned}
\end{equation} 
of the activator-inhibitor model \eqref{eq:Activator}, \eqref{eq:Inhibitor} by setting $\theta_0 = D_U$, $\theta_1 = k_1u_0\bar a$, $\theta_2=k_1$, $\theta_3 = k_2\epsilon b$, $\overline F=0$ for some $\bar a>0$ and
\begin{align}
	F_1(U) & = -\frac{a(|U|_{L^2})}{\bar a} U (u_0-U) , 
	\label{eq:FNmodelF1} \\
	F_2(U) & = U^2(u_0-U),\label{eq:FNmodelF2} \\
	F_3(U)(t) & = -\int_0^te^{(t-r) 
	(D_V\Delta-\epsilon I)} U_r \mathrm{d}r. 
	\label{eq:FNmodelF3}
\end{align}
Here $e^{D_V\Delta-\epsilon I}$ is the semigroup generated by 
$D_V\Delta-\epsilon I$. Note that $F_3$ now depends on the 
whole trajectory of $U$, so that the resulting stochastic 
evolution equation \eqref{eq:SPDEmodelLinearFN} is no longer 
Markovian.  \\

For the activator-inhibitor system \eqref{eq:Activator}, \eqref{eq:Inhibitor}, we can verify well-posedness directly. For completeness, we state the optimal regularity results for both $U$ and $V$, but our main focus lies on the observed variable $X=U$. 

\begin{proposition}\label{prop:FNWellPosed}
	Let $\gamma>d/4$. Then there is a unique solution $(U,V)$ to \eqref{eq:Activator}, \eqref{eq:Inhibitor}. Furthermore, $U$ satisfies $(A_s)$ for any $s<2\gamma-d/2+1$, and $V$ satisfies $(A_s)$ for any $s<2\gamma-d/2+3$.
\end{proposition}
The proof is deferred to Appendix \ref{app:FNWellPosed}. 

\subsubsection{Basic Regularity Results}\label{sec:BasicRegularityResults}

The nonlinear term $F$ is assumed to satisfy (cf. \cite{AltmeyerCialencoPasemann20}):
\begin{itemize}
	\item[$(F_{s,\eta})$] There is $b>0$ and $\epsilon>0$ such that
		\begin{align*}
			|(-\Delta)^{\frac{s-2+\eta+\epsilon}{2}}F_{\theta_{1:K}}(Y)|_{C(0,T;H)} \leq c(\theta_{1:K})(1+|(-\Delta)^{\frac{s}{2}}Y|_{C(0,T;H)})^b
		\end{align*}
		for $Y\in C(0,T;H^s)$, where $c$ depends continuously on $\theta_{1:K}$. 
\end{itemize}
In particular, if $F(Y)(t) = F(Y_t)$, this simplifies to 
\begin{align}
	|F_{\theta_{1:K}}(Y)|_{s-2+\eta+\epsilon} \leq c(\theta_{1:K})(1+|Y|_{s})^b
\end{align}
for $Y\in H^s$. In order to control the regularity of $X$, we apply a splitting argument (see also \cite{CialencoGlattHoltz11, PasemannStannat20, AltmeyerCialencoPasemann20}) and write $X=\overline X+\widetilde X$, where $\overline X$ is the solution to the linear SPDE
\begin{align}\label{eq:SplittingEquationLinear}
	\mathrm{d}\overline X_t = \theta_0\Delta \overline X_t\mathrm{d}t + B\mathrm{d}W_t,\quad \overline X_0 = 0,
\end{align}
where $W$ is the same cylindrical Wiener process as in \eqref{eq:SPDEmodelGeneral}, and $\widetilde X$ solves a random PDE of the form
\begin{align}\label{eq:SplittingEquationNonlinear}
	\mathrm{d}\widetilde X = (\theta_0\Delta \widetilde X_t + F_{\theta_{1:K}}(\overline X+\widetilde X)(t))\mathrm{d}t,\quad\overline X_0 = X_0.
\end{align}

\begin{lemma}\label{lem:LinearRegularity}
	The process $\overline X$ is Gaussian, and for any $p\geq 1$, $s<2\gamma-d/2+1$:
	\begin{align}
		\mathbb{E}\left[\sup_{0\leq t\leq T}|\overline X_t|_s^p\right] < \infty.
	\end{align}
\end{lemma}
\begin{proof}
	This is classical, see e.g. \cite{DaPratoZabczyk14, LiuRockner15}.
\end{proof}

\begin{proposition}\label{prop:AdditionalRegularity} \
	\begin{enumerate}
		\item Let $(A_{s})$ and $(F_{s, \eta})$ hold. Then for any $p\geq 1$:
		\begin{align}\label{eq:BoundNonlinearProcess}
			\mathbb{E}\left[\sup_{0\leq t\leq T}|\widetilde X_t|_{s+\eta}^p\right]<\infty.
		\end{align}
		In particular, if $s+\eta<2\gamma-d/2+1$, then $(A_{s+\eta})$ is true.
		\item Let $G:C(0,T;H)\supset D(G)\rightarrow C(0,T;H)$ be any function such that $(F_{s,\eta})$ holds for $G$. Then for $s < 2\gamma-d/2+1$ and $p\geq 1$:
		\begin{align}
			\mathbb{E}\left[\sup_{0\leq t\leq T}|G(X)(t)|_{s+\eta-2}^p\right] < \infty.
		\end{align}
		In particular, 
		\begin{align}
			\mathbb{E}\int_0^T|G(X)(t)|_{s+\eta-2}^2\mathrm{d}t < \infty.
		\end{align}
	\end{enumerate}
\end{proposition}
\begin{proof} \
	\begin{enumerate}
		\item For $t\in [0,T]$ and $\epsilon>0$,
			\begin{align*}
				|\widetilde X_t|_{s+\eta}
					&\leq |S(t)X_0|_{s+\eta} + \int_0^t|S(t-r)F_{\theta_{1:K}}(X)(t)|_{s+\eta}\mathrm{d}r \\
					&\leq |X_0|_{s+\eta} + \int_0^t(t-r)^{-1+\epsilon/2}|F_{\theta_{1:K}}(X)(t)|_{s-2+\eta+\epsilon}\mathrm{d}r \\
					&\leq |X_0|_{s+\eta} + \frac{2}{\epsilon}T^\frac{\epsilon}{2}\sup_{0\leq t\leq T}|F_{\theta_{1:K}}(X)(t)|_{s-2+\eta+\epsilon} \\
					&\leq |X_0|_{s+\eta} + \frac{2}{\epsilon}T^\frac{\epsilon}{2}c(\theta_{1:K})(1+|X|_{C(0,T;H^s)})^b,
			\end{align*}
			where $\theta_1,\dots,\theta_K$ are the true parameters. This implies \eqref{eq:BoundNonlinearProcess}. If $s+\eta < 2\gamma-d/2+1$, then a bound as in \eqref{eq:BoundNonlinearProcess} holds for $\overline X$ by Lemma \ref{lem:LinearRegularity}, and the claim follows. 
		\item This follows from 
			\begin{align}
				\mathbb{E}\left[\sup_{0\leq t\leq T}|G(X)(t)|_{s+\eta-2}^p\right] &\leq c\mathbb{E}\left[\left(1+\sup_{0\leq t\leq T}|X_t|_s\right)^{bp}\right] < \infty.
			\end{align}
	\end{enumerate}
\end{proof}

\subsection{Statistical Inference: The General Model}

The projected process $P_NX$ induces a measure $\mathbb{P}_\theta$ on $C(0,T;\R^N)$. Heuristically (see \cite[Section 7.6.4]{LiptserShiryayev77}), we have the following representation for the density with respect to $\mathbb{P}_{\overline \theta}^N$ for an arbitrary reference parameter $\overline\theta\in\Theta$:

\begin{align*}
	\frac{\mathrm{d}\mathbb{P}^N_\theta}{\mathrm{d}\mathbb{P}^N_{\overline\theta}}(X^N) &= \exp\left(-\frac{1}{\sigma^2}\int_0^T\left\langle(\theta_0-\overline\theta_0)\Delta X^N_t,(-\Delta)^{2\gamma}\mathrm{d}X^N_t\right\rangle\right. \\
	    & \quad -\frac{1}{\sigma^2}\int_0^T\left\langle P_N (F_{\theta_{1:K}}- F_{\overline\theta_{1:K}})(X),(-\Delta)^{2\gamma}\mathrm{d}X^N_t\right\rangle \\
		& \quad\left.+\frac{1}{2\sigma^2}\int_0^T\left\langle(\theta_0-\overline\theta_0)\Delta X^N_t + P_N(F_{\theta_{1:K}}-F_{\overline\theta_{1:K}})(X), \right.\right.\\
		& \quad\quad \left.\left.(-\Delta)^{2\gamma}\left[(\theta_0+\overline\theta_0)\Delta X^N_t + P_N (F_{\theta_{1:K}}+F_{\overline\theta_{1:K}})(X)\right]\right\rangle\mathrm{d}t\right).
\end{align*}

By setting the score (i.e. the gradient with respect to $\theta$ of the log-likelihood) to zero, and by formally substituting the (fixed) parameter $\gamma$ by a (free) parameter $\alpha$, we get the following maximum likelihood equations:
\begin{align*}
	\hat\theta_0^N\int_0^T|(-\Delta)^{1+\alpha}X^N_t|_H^2\mathrm{d}t 
	    &= \int_0^T\langle(-\Delta)^{1+2\alpha}X^N_t,P_NF_{\hat\theta_{1:K}^N}(X)\rangle\mathrm{d}t \hspace{-5cm} & \\
		&\quad -\int_0^T\langle(-\Delta)^{1+2\alpha}X^N_t,\mathrm{d}X^N_t\rangle, \\
	-\hat\theta_0^N\int_0^T\langle(-\Delta)^{1+2\alpha}X^N_t,\partial_{\theta_i}P_NF_{\hat\theta_{1:K}^N}(X)\rangle\mathrm{d}t \hspace{-2.5cm} & \\
	    &\hspace{-2cm}= - \int_0^T\langle(-\Delta)^{2\alpha}P_NF_{\hat\theta_{1:K}^N}(X),\partial_{\theta_i}P_NF_{\hat\theta_{1:K}^N}(X)\rangle\mathrm{d}t \hspace{-10cm} & \\
		&\hspace{-2cm}\quad + \int_0^T\langle(-\Delta)^{2\alpha}\partial_{\theta_i}P_NF_{\hat\theta_{1:K}^N}(X),\mathrm{d}X^N_t\rangle.
\end{align*}

Any solution $(\hat\theta^N_0,\dots,\hat\theta^N_K)$ to these equations is a (joint) maximum likelihood estimator (MLE) for $(\theta_0,\dots,\theta_K)$. W.l.o.g. we assume that the MLE is unique, otherwise fix any solution. We are interested in the asymptotic behavior of this estimator as $N\rightarrow\infty$, i.e. as more and more spatial information (for fixed $T>0$) is available. While identifiability of $\theta_1,\dots,\theta_K$ in finite time depends in general on additional structural assumptions on $F$, the diffusivity $\theta_0$ is expected to be identifiable in finite time under mild assumptions. Indeed, the argument is similar to \cite{CialencoGlattHoltz11, PasemannStannat20}, but we have to take into account the dependence of $\hat\theta^N_0$ on the other estimators $\hat\theta^N_1,\dots,\hat\theta^N_K$. 
Note that the likelihood equations give the following useful representation for $\hat\theta^N_0$:
\begin{align}\label{eq:EstimatorClassicalRepresentation}
	\hat\theta^N_0 = \frac{-\int_0^T\langle(-\Delta)^{1+2\alpha}X^N_t,\mathrm{d}X^N_t\rangle + \int_0^T\langle(-\Delta)^{1+2\alpha}X^N_t,P_NF_{\hat\theta_{1:K}^N}(X)\rangle\mathrm{d}t}{\int_0^T|(-\Delta)^{1+\alpha}X^N_t|_H^2\mathrm{d}t}.
\end{align}

By plugging in the dynamics of $X$ according to \eqref{eq:SPDEmodelGeneral}, we obtain the following decomposition:
\begin{align}
    \hat\theta^N_0 - \theta_0 &= \frac{\int_0^T\langle(-\Delta)^{1+2\alpha}X^N_t,P_NF_{\hat\theta^N_{1:K}}(X)\rangle\mathrm{d}t}{\int_0^T|(-\Delta)^{1+\alpha}X^N_t|_H^2\mathrm{d}t} \label{eq:EstimatorDecompositionWithBias}\\
    &\quad\hspace{-1.5cm} - \frac{\int_0^T\langle(-\Delta)^{1+2\alpha}X^N_t,P_NF_{\theta_{1:K}}(X)\rangle\mathrm{d}t}{\int_0^T|(-\Delta)^{1+\alpha}X^N_t|_H^2\mathrm{d}t} - \frac{\int_0^T\langle(-\Delta)^{1+2\alpha}X^N_t,B\mathrm{d}W^N_t\rangle}{\int_0^T|(-\Delta)^{1+\alpha}X^N_t|_H^2\mathrm{d}t}. \nonumber
\end{align}
The right-hand side vanishes whenever for large $N$ the denominator grows faster than the numerator in each of the three fractions. In principle, strong oscillation of the reaction parameter estimates $\hat\theta^N_{1:K}$ may influence the convergence rate for the first term, so in order to exclude undesirable behaviour, we assume that $\hat\theta^N_{1:K}$ is bounded in probability\footnote{A sequence of estimators $(\hat\theta^N)_{N\in\N}$ is called bounded in probability (or tight) if $\sup_{N\in\N}\mathbb{P}(|\hat\theta^N|>M)\rightarrow 0$ as $M\rightarrow\infty$.}. This is a mild assumption which is in particular satisfied if the estimators for the reaction parameters are consistent. 
In Section \ref{sec:StatInfLinearModel} we verify this condition for the case that $F$ depends linearly on $\theta_{1:K}$. 
Regarding the third term, we exploit the martingale structure of the noise in order to capture the growth in $N$. Different noise models may be used in \eqref{eq:SPDEmodelGeneral} without changing the result, as long as the numerator grows slower than the denominator. For example, the present argument directly generalizes to noise of martingale type\footnote{
The generalization of our results to the case of multiplicative noise depends crucially on the noise model, see e.g. \cite{CialencoLototsky09, Cialenco10}. This is beyond the scope of the present work. 
}. 
Now, the growth of the denominator can be quantified as follows:

\begin{lemma}\label{lem:DenominatorRate}
	Let $\alpha>\gamma-d/4-1/2$, let further $\eta,s_0>0$ such that $(A_s)$ and $(F_{s,\eta})$ are true for $s_0\leq s < 2\gamma+1-d/2$. Then 
	\begin{align}
		\int_0^T|(-\Delta)^{1+\alpha}X^N_t|_H^2\mathrm{d}t
			&\asymp \mathbb{E}\int_0^T|(-\Delta)^{1+\alpha}\overline X^N_t|_H^2\mathrm{d}t \\
			&\asymp C_\alpha N^{\frac{2}{d}(2\alpha-2\gamma+1)+1}
	\end{align}
	in probability, with
	\begin{align}
		C_\alpha = \frac{T\Lambda^{2\alpha-2\gamma+1}d}{2\theta(4\alpha-4\gamma+2+d)}.
	\end{align}
\end{lemma}
\begin{proof}
	Using Proposition \ref{prop:AdditionalRegularity} (i), the proof is exactly as in \cite[Proposition 4.6] {PasemannStannat20}. 
\end{proof}

\begin{theorem}\label{thm:DiffusivityEstimatorProperties}
	Assume that the likelihood equations are solvable for $N\geq N_0$, assume that $(\hat\theta^N_i)_{N\geq N_0}$ is bounded in probability for $i=1,\dots,K$. Let $\alpha>\gamma-d/4-1/2$ and $\eta, s_0>0$ such that $(A_s)$ and $(F_{s, \eta})$ hold
	for any $s_0\leq s < 2\gamma+1-d/2$. 
	Then the following is true:
	\begin{enumerate}
		\item $\hat\theta^N_0$ is a consistent estimator for $\theta_0$, i.e. $\hat\theta^N_0\xrightarrow{\mathbb{P}}\theta_0$. 
		\item If $\eta\leq 1 + d/2$, then $N^{r}(\hat\theta^N_0-\theta_0)\xrightarrow{\mathbb{P}}0$ for any $r<\eta/d$. 
		\item If $\eta > 1+d/2$, then
		\begin{align}\label{eq:AsymptoticNormality}
			N^{\frac{1}{2}+\frac{1}{d}}(\hat\theta^N_0-\theta_0)\xrightarrow{d}\mathcal{N}(0, V),
		\end{align}
		with $V = 2\theta_0(4\alpha-4\gamma+d+2)^2 / (Td\Lambda^{2\alpha-2\gamma+1}(8\alpha-8\gamma+d+2))$.
	\end{enumerate}
\end{theorem}

\begin{proof}
	By means of 
	the decomposition \eqref{eq:EstimatorDecompositionWithBias}, 
	we proceed as in 
	\cite{PasemannStannat20}. Denote by $\hat\theta^{\mathrm{full},N}_0$ the estimator which is given by \eqref{eq:EstimatorClassicalRepresentation} if the $\hat\theta^N_1,\dots,\hat\theta^N_K$ are substituted by the true values $\theta_1,\dots,\theta_K$. 
	In this case, the estimation error simplifies to 
	\begin{align}
		\hat\theta^{\mathrm{full},N}_0 - \theta_0 = -c_N \frac{\int_0^T\langle(-\Delta)^{1+2\alpha-\gamma}X^N_t,\mathrm{d}W^N_t\rangle}{\sqrt{\int_0^T|(-\Delta)^{1+2\alpha-\gamma}X^N_t|_H^2\mathrm{d}t}}
	\end{align}
	with   
	$$ 
	c_N = \frac{\sqrt{\int_0^T |(-\Delta)^{1+2\alpha-\gamma} 
	X^N_t|_H^2\mathrm{d}t}} 
	{\int_0^T|(-\Delta)^{1+\alpha}X^N_t|_H^2\mathrm{d}t}. 
	$$ 
	By Lemma \ref{lem:DenominatorRate}, the rescaled 
	prefactor $c_N N^{1/2+1/d}$ converges in probability to 
	$\sqrt{C_{2\alpha-\gamma}} / C_\alpha$. The second 
	factor converges in distribution to a standard normal 
	distribution $\mathcal{N}(0,1)$ by the central limit 
	theorem for local martingales  
	(see \cite[Theorem 5.5.4 (I)] {LiptserShiryayev89}, 
	\cite[Theorem VIII.4.17]{JacodShiryayev03}). This proves \eqref{eq:AsymptoticNormality} for $\hat\theta^{\mathrm{full},N}_0$. To conclude, we bound the bias term depending on $\hat\theta^N_1,\dots,\hat\theta^N_K$ as follows, using $|P_NY|_{s_2}\leq \lambda_N^{(s_2-s_1)/2}|P_NY|_{s_1}$ for $s_1<s_2$: Let $\delta>0$. Then
	\begin{align*}
		\int_0^T\langle(-\Delta)^{1+2\alpha}X^N_t,P_NF_{\hat\theta_{1:K}^N}(X)\rangle\mathrm{d}t \hspace{-5cm} & \\
			&\leq \left(\int_0^T|(-\Delta)^{1+\alpha}X^N_t|_H^2\mathrm{d}t\right)^\frac{1}{2}\left(\int_0^T|(-\Delta)^{\alpha}P_NF_{\hat\theta_{1:K}^N}(X)|_H^2\mathrm{d}t\right)^\frac{1}{2} \\
			&\lesssim N^{\frac{1}{d}(2\alpha-2\gamma+1)+\frac{1}{2}}\left(\int_0^T|(-\Delta)^{\alpha}P_NF_{\hat\theta_{1:K}^N}(X)|_H^2\mathrm{d}t\right)^\frac{1}{2} \\
			&\lesssim N^{\frac{2}{d}(2\alpha-2\gamma+1)+1 - \frac{\eta-\delta}{d}}\left(\int_0^T|(-\Delta)^{\gamma + \frac{1}{2}-\frac{d}{4}-1+\frac{\eta-\delta}{2}}P_NF_{\hat\theta_{1:K}^N}(X)|_H^2\mathrm{d}t\right)^\frac{1}{2},
	\end{align*}
	so using $(F_{2\gamma+1-d/2-\delta,\eta})$, 
	\begin{align*}
		\frac{\int_0^T\langle(-\Delta)^{1+2\alpha}X^N_t,P_NF_{\hat\theta_{1:K}^N}(X)\rangle\mathrm{d}t}{\int_0^T|(-\Delta)^{1+\alpha}X^N_t|_H^2\mathrm{d}t} \lesssim c(\hat\theta^N_{1:K})N^{-(\eta-\delta) / d}
	\end{align*}
	As $c(\hat\theta^N_{1:K})$ is bounded in probability and $\delta>0$ is arbitrarily small, the claim follows. The remaining term involving the true parameters $\theta_1,\dots,\theta_K$ is similar. This concludes the proof. 
\end{proof}

It is clear that a Lipschitz condition on $F$ with respect to $\theta_{1:K}$ allows to bound $\hat\theta^N_0-\thetafull$ in terms of $\lvert\hat\theta^N_{1:K}-\theta_{1:K}\rvert N^{-(\eta-\delta)/d}$ for $\delta>0$, using the notation from the previous proof. In this case, consistency of $\hat\theta^N_i$, $i=1,\dots,K$, may improve the rate of convergence of $\hat\theta^N_0$. However, as noted before, in general we cannot expect $\hat\theta^N_i$, $i=1,\dots,K$, to be consistent as $N\rightarrow\infty$. 

\subsection{Statistical Inference: The Linear Model}\label{sec:StatInfLinearModel}

We put particular emphasis on the case that the nonlinearity $F$ depends linearly on its parameters:
\begin{align}\label{eq:SPDEmodelLinear}
	\mathrm{d}X_t = \left(\theta_0\Delta X_t + \sum_{i=1}^K\theta_iF_i(X) + \overline F(X)\right)\mathrm{d}t + B\mathrm{d}W_t.
\end{align}
This model includes the FitzHugh--Nagumo system in the form \eqref{eq:SPDEmodelLinearFN}.
We state an additional verifiable condition, depending on the contrast parameter $\alpha\in\R$, which guarantees that the likelihood equations are well-posed, amongst others.

\begin{itemize}
	\item[$(L_\alpha)$] The terms $F_1(Y),\dots,F_K(Y)$ are well-defined as well as linearly independent in $L^2(0,T;H^{2\alpha})$ for every non-constant $Y\in C(0,T;C(\mathcal{D}))$.
\end{itemize}
In particular, condition $(L_\alpha)$ implies for $i=1,\dots,K$ that
\begin{align}
	\int_0^T|(-\Delta)^\alpha F_i(X)|_H^2\mathrm{d}t > 0.
\end{align}

For linear SPDEs, similar considerations have been made first in \cite[Chapter 3]{Huebner93}.
The maximum likelihood equations for the linear model \eqref{eq:SPDEmodelLinear} simplify to 

\begin{align}
	A_N(X)\hat\theta^N(X) = b_N(X),
\end{align}
where
\begin{align*}
	A_N(X)_{0,0} &= \int_0^T|(-\Delta)^{1+\alpha}X^N_t|_H^2\mathrm{d}t, \\
	A_N(X)_{0,i} = A_N(X)_{i,0} &= -\int_0^T\langle(-\Delta)^{1+2\alpha}X^N_t,P_NF_i(X)\rangle\mathrm{d}t, \\
	A_N(X)_{i,j} &= \int_0^T\langle (-\Delta)^{2\alpha}P_NF_i(X),P_NF_j(X)\rangle\mathrm{d}t
\end{align*}
for $i,j=1,\dots,K$, and
\begin{align*}
	b_N(X)_0 &= -\int_0^T\langle(-\Delta)^{1+2\alpha}X^N_t,\mathrm{d}X^N_t\rangle \\
	         &\quad + \int_0^T\langle(-\Delta)^{1+2\alpha}X^N_t,P_N\overline F(X)\rangle\mathrm{d}t, \\
	b_N(X)_i &= \int_0^T\langle(-\Delta)^{2\alpha}P_NF_i(X),\mathrm{d}X^N_t\rangle \\
	         &\quad - \int_0^T\langle(-\Delta)^{2\alpha}P_NF_i(X),P_N\overline F(X)\rangle\mathrm{d}t
\end{align*}
for $i=1,\dots,K$. \\

In order to apply Theorem \ref{thm:DiffusivityEstimatorProperties}, we 
need that the estimators $\hat\theta^N_1,\dots,\hat\theta^N_K$ are bounded in probability. 
\begin{proposition}\label{prop:LinearModelEstimatorBounded}
	In the setting of this section, let $\eta,s_0>0$ such that $(A_s)$ and $(F_{s,\eta})$ are true for $s_0\leq s < 2\gamma+1-d/2$. For $\gamma-d/4-1/2<\alpha\leq \gamma\wedge(\gamma-d/4-1/2+\eta/2)\wedge(\gamma-d/8-1/4+\eta/4)$, let $(L_\alpha)$ be true. Then the $\hat\theta^N_i$, $i=0,\dots,K$, are bounded in probability.
\end{proposition}

The proof of Proposition \ref{prop:LinearModelEstimatorBounded} is given in Appendix \ref{app:LinearModelEstimatorBounded}. 
We note that the upper bound on $\alpha$ can be relaxed in general, depending on the exact asymptotic behaviour of $A_N(X)_{i,i}$, $i=1,\dots,K$. Proposition \ref{prop:LinearModelEstimatorBounded} together with Theorem \ref{thm:DiffusivityEstimatorProperties} gives conditions for $\hat\theta^N_0$ to be consistent and asymptotically normal in the linear model \eqref{eq:SPDEmodelLinear}. 
In particular, we immediately get for the activator-inhibitor model \eqref{eq:Activator}, \eqref{eq:Inhibitor}, as the linear independency condition $(L_\alpha)$ is trivially satisfied and $\eta$ can be chosen arbitrarily close to $2$:
\begin{theorem}
	Let $\gamma>d/4$. Then $\hat\theta^N_0$ has the following properties in the activator-inhibitor model \eqref{eq:Activator}, \eqref{eq:Inhibitor}:
	\begin{enumerate}
		\item In $d=1$, let 
		$\gamma-3/4<\alpha \leq\gamma$. 
		Then $\hat\theta^N_0$ is a consistent estimator for $\theta_0$, which is asymptotically normal as in \eqref{eq:AsymptoticNormality}.
		\item In $d=2$, let $\gamma-1<\alpha < \gamma$. Then $\hat\theta^N_0$ is a consistent estimator for $\theta_0$ with optimal convergence rate, i.e. $N^r(\hat\theta^N_0-\theta_0)\xrightarrow{\mathbb{P}}0$ for any $r<1$. 
	\end{enumerate}
\end{theorem}

\section{Application to Activator-Inhibitor Models of Actin Dynamics}
\label{sec:data}

The actin cytoskeleton is a dense polymer meshwork at the inner face of the plasma membrane that determines the shape and mechanical stability of a cell.
Due to the continuous polymerization and depolymerization of the actin filaments, it displays a dynamic network structure that generates complex spatiotemporal patterns.
These patterns are the basis of many essential cellular functions, such as endocytic processes, cell shape changes, and cell motility~\cite{blanchoin_actin_2014}.
The dynamics of the actin cytoskeleton is controlled and guided by upstream signaling pathways, which are known to display typical features of nonequilibrium systems, such as oscillatory instabilities and the emergence of traveling wave patterns~\cite{devreotes_excitable_2017,beta_intracellular_2017}.
Here we use giant cells of the social amoeba {\it D.~discoideum} that allow us to observe these cytoskeletal patterns over larger spatial domains~\cite{Gerhardt_actin_2014}.
Depending on the genetic background and the developmental state of the cells, different types of patterns emerge in the cell cortex.
In particular, pronounced actin wave formation is observed as the consequence of a mutation in the upstream signaling pathway --- a knockout of the RasG-inactivating RasGAP NF1 --- which is present for instance in the commonly used laboratory strain AX2~\cite{veltman_plasma_2016}.
\begin{figure}[t]
	\includegraphics[width=0.98\textwidth]{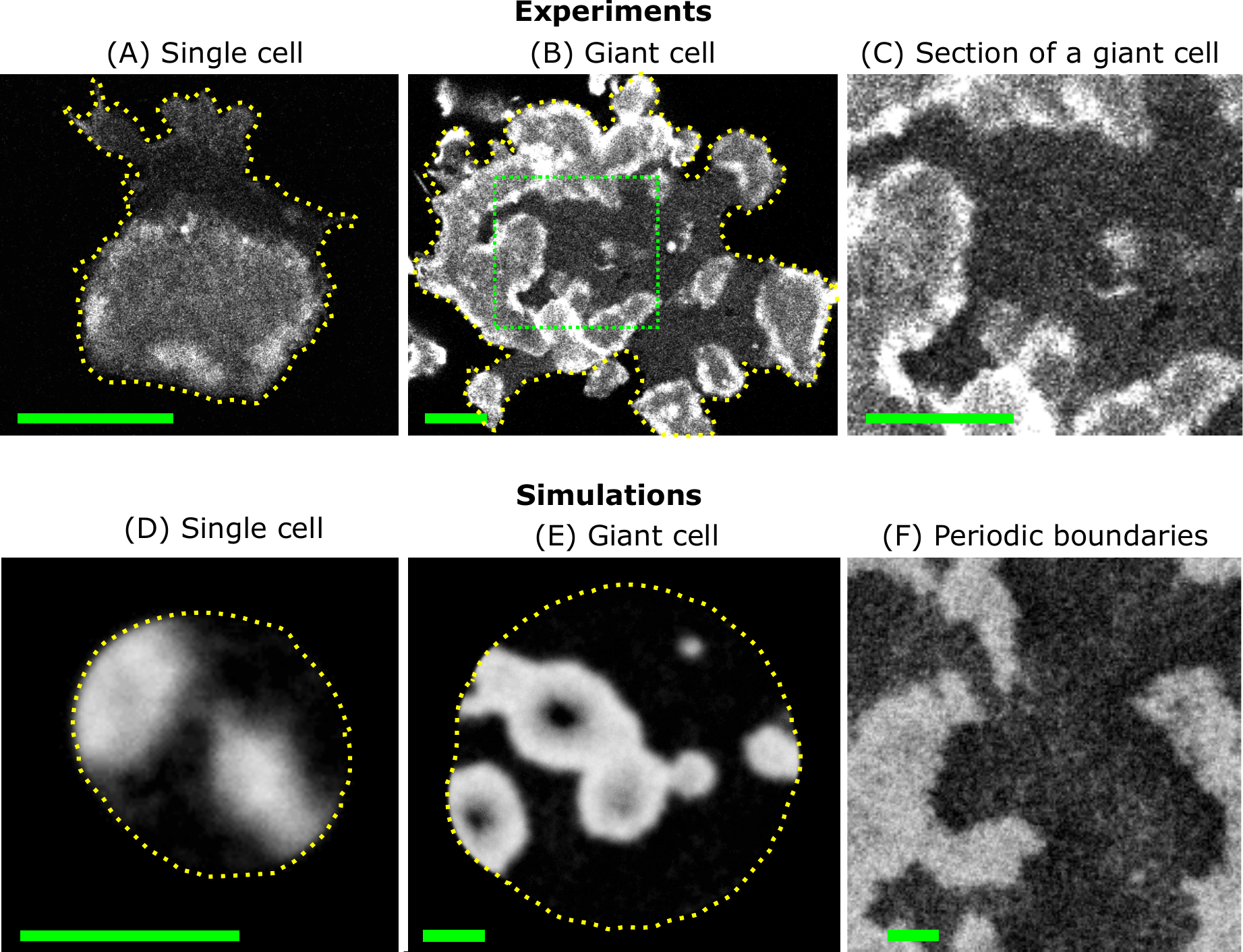}
	\caption{Actin waves in experiments (top) and model simulations (bottom).
	(A)~Normal-sized cell with a circular actin wave.
	(B)~Giant cell with several fragmented actin waves.
	(C)~Subsection of the cortical area of the giant cell shown in (B), indicated as dotted green rectangle.
	Experimental images are confocal microscopy recordings of mRFP-LimE$\Delta$ expressing {\it D.~discoideum} AX2 cells, see~\cite{Gerhardt_actin_2014}.
	(Bottom) Simulations of the stochastic reaction-diffusion model \eqref{eq:IntroActivator}, \eqref{eq:IntroInhibitor} in a (D)~small and (E)~large domain, defined by a dynamically evolving phase field and (F)~with periodic boundary conditions.
	For details on the phase field simulations, see~\cite{FlemmingFontAlonsoBeta20}.
	(Scale bars, 10~\si{\micro}m)
	Details on the numerical implementation can be found in Appendix \ref{app:methods}.
	\label{fig:Experiment_and_Simulations}}
\end{figure}
Giant cells of NF1-deficient strains thus provide a well-suited setting to study the dynamics of actin waves and their impact on cell shape and division~\cite{FlemmingFontAlonsoBeta20}.

In Figure~\ref{fig:Experiment_and_Simulations}A and B, we show a normal-sized and a giant {\it D.~discoideum} cell in the wave-forming regime for comparison.
Images were recorded by confocal laser scanning microscopy and display the distribution of mRFP-LimE$\Delta$, a fluorescent marker for filamentous actin, in the cortex at the substrate-attached bottom membrane of the cell.
As individual actin filaments are not resolved by this method, the intensity of the fluorescence signal reflects the local cortical density of filamentous actin.
Rectangular subsections of the inner part of the cortex of giant cells as displayed in panel~(C) were used for data analysis in Section~\ref{sec:data_analysis}. 

Many aspects of subcellular dynamical patterns have been addressed by reaction-diffusion models.
While some models rely on detailed modular approaches~\cite{beta_bistable_2008,devreotes_excitable_2017}, others have focused on specific parts of the upstream signaling pathways, such as the phos\-pha\-ti\-dyl\-i\-no\-si\-tol lipid signaling system~\cite{arai_self-organization_2010} or Ras signaling \cite{fukushima_excitable_2019}.
To describe wave patterns in the actin cortex of giant {\it D.~discoideum} cells, the noisy FitzHugh--Nagumo type reaction-diffusion system \eqref{eq:IntroActivator}, \eqref{eq:IntroInhibitor}, combined with a dynamic phase field, has been recently proposed~\cite{FlemmingFontAlonsoBeta20}.
In contrast to the more detailed biochemical models, the structure of this model is rather simple.
Waves are generated by noisy bistable/exci\-ta\-ble kinetics with an additional control of the total amount of activator $U$.
This constraint dynamically regulates the amount of $U$ around a constant level in agreement with the corresponding biological restrictions.
Elevated levels of the activator represent typical cell front markers, such as active Ras, PIP3, Arp2/3, and freshly polymerized actin that are also concentrated in the inner part of actin waves.
On the other hand, markers of the cell back, such as PIP2, myosin II, and cortexillin, correspond to low values of $U$ and are found outside the wave area~\cite{schroth-diez_propagating_2009}.
Tuning of the parameter $b$ allows to continuously shift from bistable to excitable dynamics, both of which are observed in experiments with {\it D.~discoideum} cells.
In Figure~\ref{fig:Experiment_and_Simulations}D-F, the results of numerical simulations of this model 
displaying excitable dynamics are shown. 
Examples for bounded domains that correspond to normal-sized and giant cells are shown, as well as results with periodic boundary conditions that were used in the subsequent analysis.

Model parameters, such as the diffusivities, are typically chosen in an {\it ad hoc} fashion to match the speed of intracellular waves with the experimental observations.
The approach introduced in Section~\ref{sec:Theory} now allows us to estimate diffusivities from data in a more rigorous manner.
On the one hand, we may test the validity of our method on {\it in silico} data of model simulations, where all parameters are predefined.
On the other hand, we can apply our method to experimental data, such as the recordings of cortical actin waves displayed in Figure~\ref{fig:Experiment_and_Simulations}C.
This will yield an estimate of the diffusivity of the activator $U$, as dense areas of filamentous actin reflect increased concentrations of activatory signaling components.
Note, however, that the estimated value of $D_U$ should not be confused with the molecular diffusivity of a specific signaling molecule.
It rather reflects an effective value that includes the diffusivities of many activatory species of the signaling network and is furthermore affected by the specific two-dimensional setting of the model that neither includes the kinetics of membrane attachment/detachment nor the three-dimensional cytosolic volume.

\section{Diffusivity Estimation on \\ Simulated and Real Data}
\label{sec:data_analysis}

In this section we apply the methods from Section \ref{sec:Theory} to synthetic data obtained from a numerical simulation and to cell data stemming from experiments as described in Section~\ref{sec:data}.
We follow the formalism from Theorem \ref{thm:DiffusivityEstimatorProperties} and perform a Fourier decomposition on each data set. Set $\phi_k(x) = \cos(2\pi k x)$ for $k\leq 0$ and $\phi_k(x)=\sin(2\pi k x)$ for $k>0$, then  $\Phi_{k,l}(x, y)=\phi_k(x/L_1)\phi_l(y/L_2)$, $k,l\in\Z$, form an eigenbasis for $-\Delta$ on the rectangular domain $\mathcal{D}=[0,L_1]\times[0,L_2]$. The corresponding eigenvalues are given by $\lambda_{k,l}=4\pi^2((k/L_1)^2+(l/L_2)^2)$. As before, we choose an ordering $((k_N,l_N))_{N\in\N}$ of the eigenvalues (excluding $\lambda_{0,0}=0$) such that $\lambda_N = \lambda_{k_N,l_N}$ is increasing. \\

In the sequel, we will use different versions of $\hat\theta^N_0$ which correspond to different model assumptions on the reaction term $F$, concerning both the effects included in the model and a priori knowledge on the parametrization. While all of these estimators enjoy the same asymptotic properties as $N\rightarrow\infty$, it is reasonable to expect that they exhibit huge qualitative differences for fixed $N\in\N$, depending on how much knowledge on the generating dynamics is incorporated. In order to describe the model nonlinearities that we presume, we use the notation $F_1,F_2,F_3$ as in \eqref{eq:FNmodelF1}, \eqref{eq:FNmodelF2}, \eqref{eq:FNmodelF3}. As a first simplification, we substitute $F_1$ by $\widetilde F_1$ given by
\begin{align}
	\widetilde F_1(U) = -U(u_0-U)
\end{align}
in all estimators below. This corresponds to an approximation of the function $a$ by an effective average value $\bar a>0$. While this clearly does not match the full model, we will see that it does not pose a severe restriction as $a(U)$ tends to stabilize in the simulation. 
Recall the explicit representation \eqref{eq:EstimatorClassicalRepresentation} of $\hat\theta^N_0$. As before, $K$ is the number of nuisance parameters appearing in the nonlinear term $F$. We construct the following estimators which capture qualitatively different model assumptions:
\begin{enumerate}
	\item The {\it linear estimator} $\thetalin$ results from presuming $K=0$ and $F=0$.
	\item The {\it polynomial} or {\it Schl\"ogl estimator} $\thetapol$, where $K=0$ and 
	\begin{align*}
		F(u)&=k_1u(u_0-u)(u-\bar au_0)\\
		    &=-k_1\bar au_0u(u_0-u)+k_1u^2(u_0-u) \\
		    &= \theta_1\widetilde F_1(u)+\theta_2F_2(u)
	\end{align*}
	for known constants $k_1,u_0,\bar a>0$, $\theta_1 = k_1u_0\bar a$, $\theta_2=k_1$. The corresponding SPDE \eqref{eq:SPDEmodelGeneral} is called stochastic Nagumo equation or stochastic Schl\"ogl equation and arises as the limiting case $\epsilon\rightarrow 0$ of the stochastic FitzHugh--Nagumo system. 
	\item The {\it full} or {\it FitzHugh--Nagumo estimator} $\thetafull$, where $K=0$ and 
	\begin{align}\label{eq:FullEstimatorNonlinearity}
	\begin{matrix*}[l]
		F(u)&=k_1u(u_0-u)(u-\bar a u_0)-k_2v\\
		    &= \theta_1\widetilde F_1(u)+\theta_2F_2(u)+\theta_2F_3(u)
	\end{matrix*}
	\end{align}
	with $\theta_1 = k_1u_0\bar a$, $\theta_2=k_1$ and $\theta_3 = k_2\epsilon b$, where $v$ is given by $v_t=\int_0^te^{(t-r)(D_V\Delta-\epsilon I)}u_r\mathrm{d}r$. As before, $k_1,k_2,u_0,\bar a,D_v,\epsilon>0$ are known. 
\end{enumerate}
Furthermore, we modify $\thetafull$ in order to estimate different subsets of model parameters at the same time. We use the notation $\hat\theta^{i,N}_0$, where $i$ is the number of simultaneously estimated parameters. More precisely, we set $\hat\theta^{1,N}_0=\thetafull$, and additionally:
\begin{enumerate}
	\item The estimator $\thetatwo$ results from $K=1$ and $F_{\theta_1}$ given by \eqref{eq:FullEstimatorNonlinearity} for known $\theta_2,\theta_3>0$. This corresponds to an unknown $\bar a$.
	\item The estimator $\thetathree$ results from $K=2$ and $F_{\theta_1,\theta_2}$ given by \eqref{eq:FullEstimatorNonlinearity}. Only $\theta_3$ is known.
	\item The estimator $\thetafour$ results from $K=3$ and $F_{\theta_1,\theta_2,\theta_3}$ given by \eqref{eq:FullEstimatorNonlinearity}. All three parameters $\theta_1,\theta_2,\theta_3$ are unknown. 
\end{enumerate}

In all estimators in this section we set 
the regularity adjustment 
$\alpha=0$. This is a reasonable choice if the driving noise in \eqref{eq:Activator}, \eqref{eq:Inhibitor} is close to white noise. \\

It is worthwhile to note that $\thetalin$ is invariant under rescaling the intensity of the data, i.e. substituting $X$ by $cX$, $c>0$. This has the advantage that we do not need to know the physical units of the data. 
In fact, the intensity of fluorescence microscopy data may vary due to different expression levels of reporter proteins within a cell population, or fluctuations in the illumination. 
While invariance under intensity rescaling 
is a desirable property, the fact that nonlinear reaction terms are not taken into account may outweigh this advantage, especially if the SPDE model is close to the true generating process of the data. This is the case for synthetic data. The discussion in Section \ref{sec:PerformanceSyntheticData} shows that even if the model specific correction terms in \eqref{eq:EstimatorClassicalRepresentation} vanish asymptotically, their effect on the estimator may be huge in the non-asymptotic regime, especially at low resolution level. However, real data may behave differently, and a detailed nonlinear model may not reveal additional information on the underlying diffusivity, see Section \ref{sec:RealData}.

\subsection{Performance on Synthetic Data}\label{sec:PerformanceSyntheticData}

\begin{figure}[t]
	\includegraphics[width=0.49\textwidth]{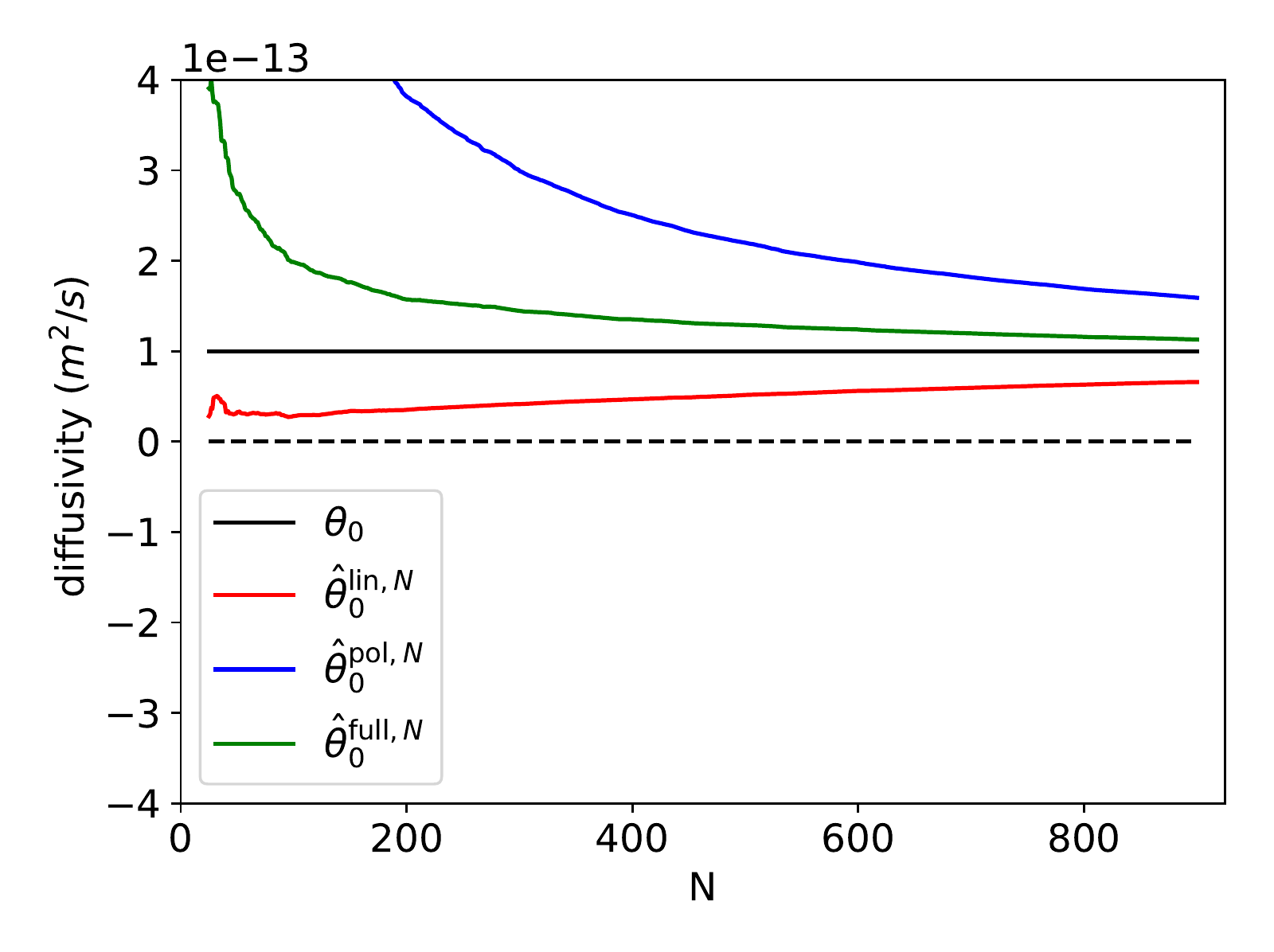}
	\includegraphics[width=0.49\textwidth]{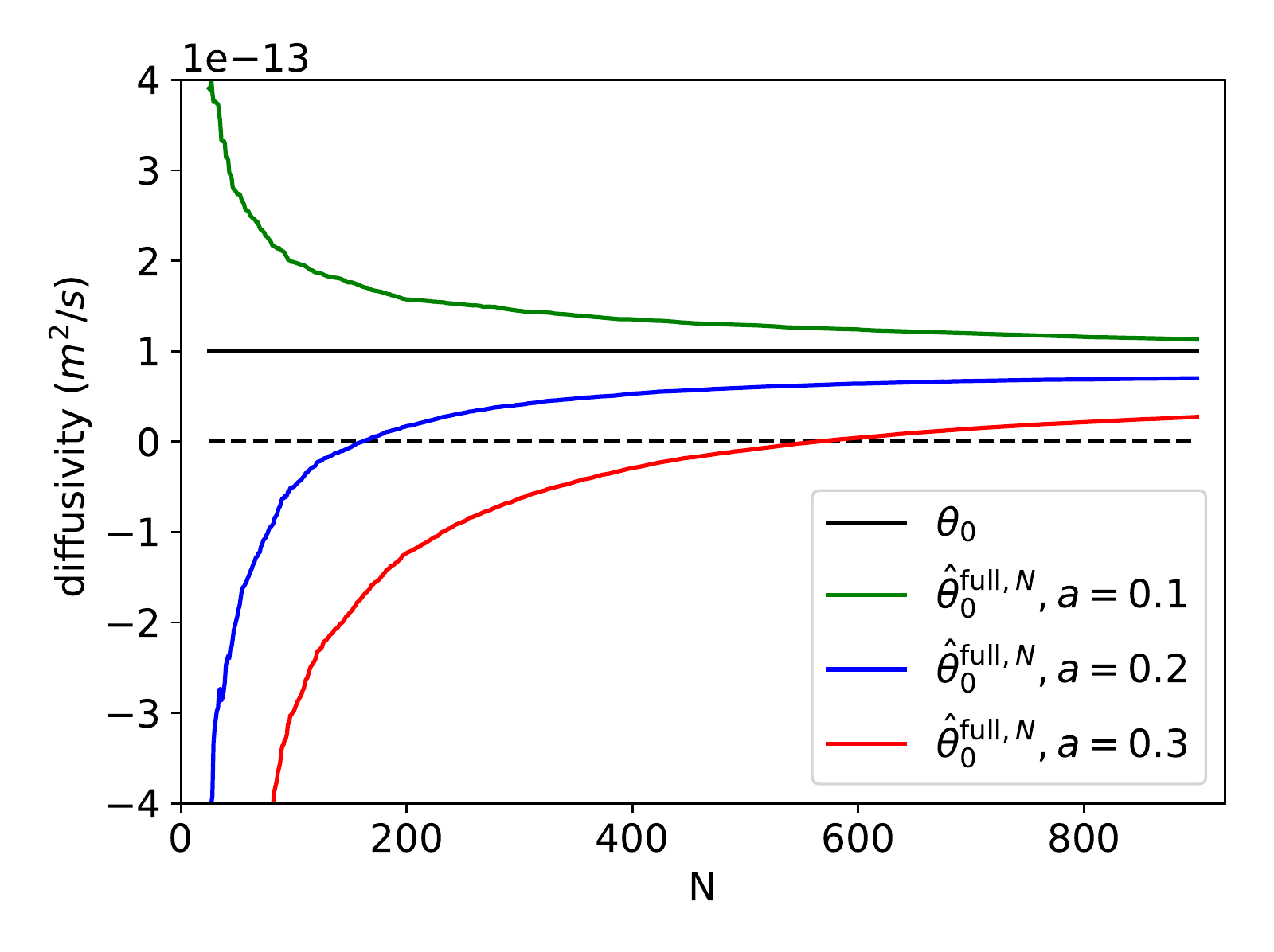}
	\includegraphics[width=0.49\textwidth]{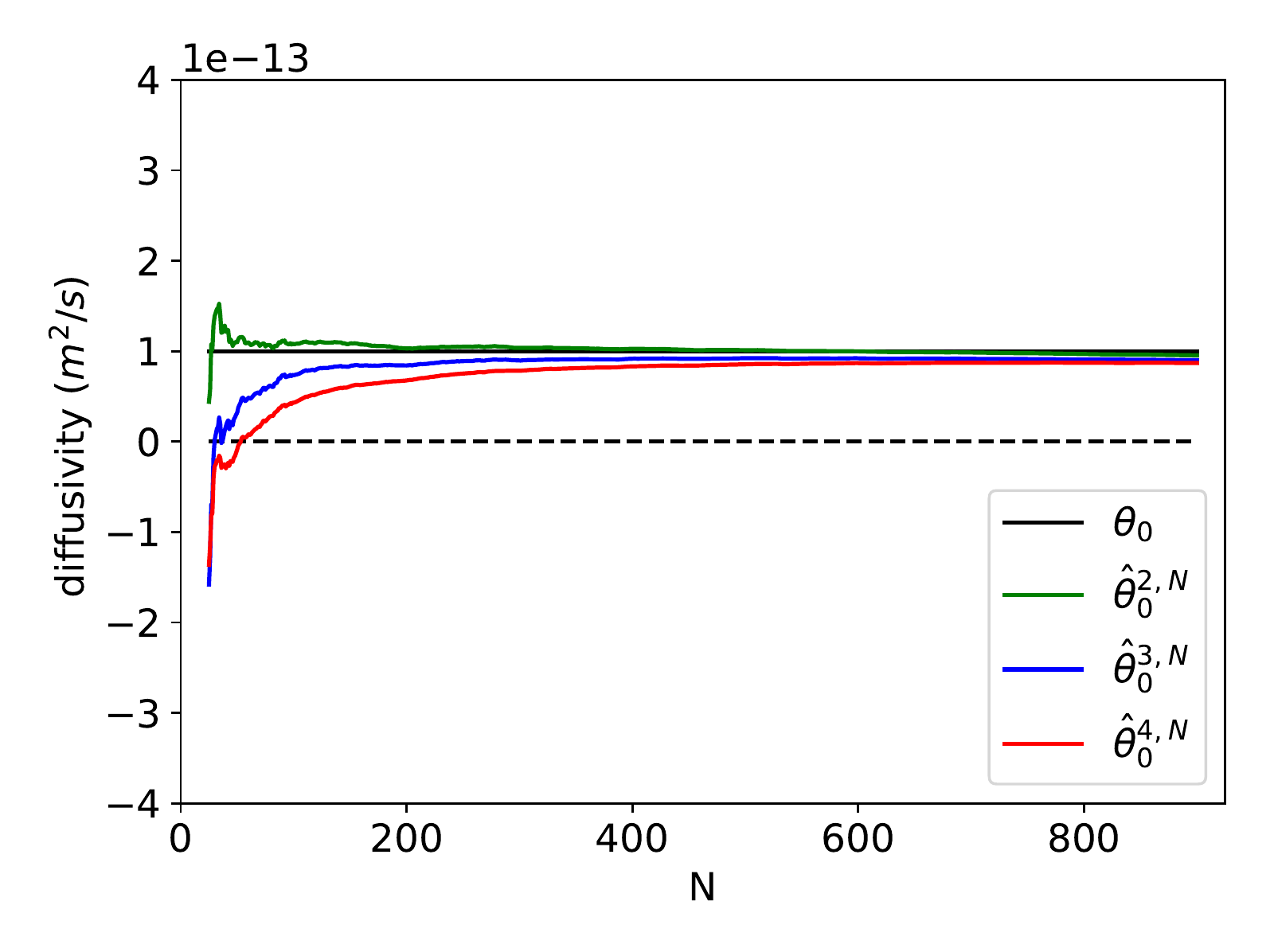}
	\includegraphics[width=0.49\textwidth]{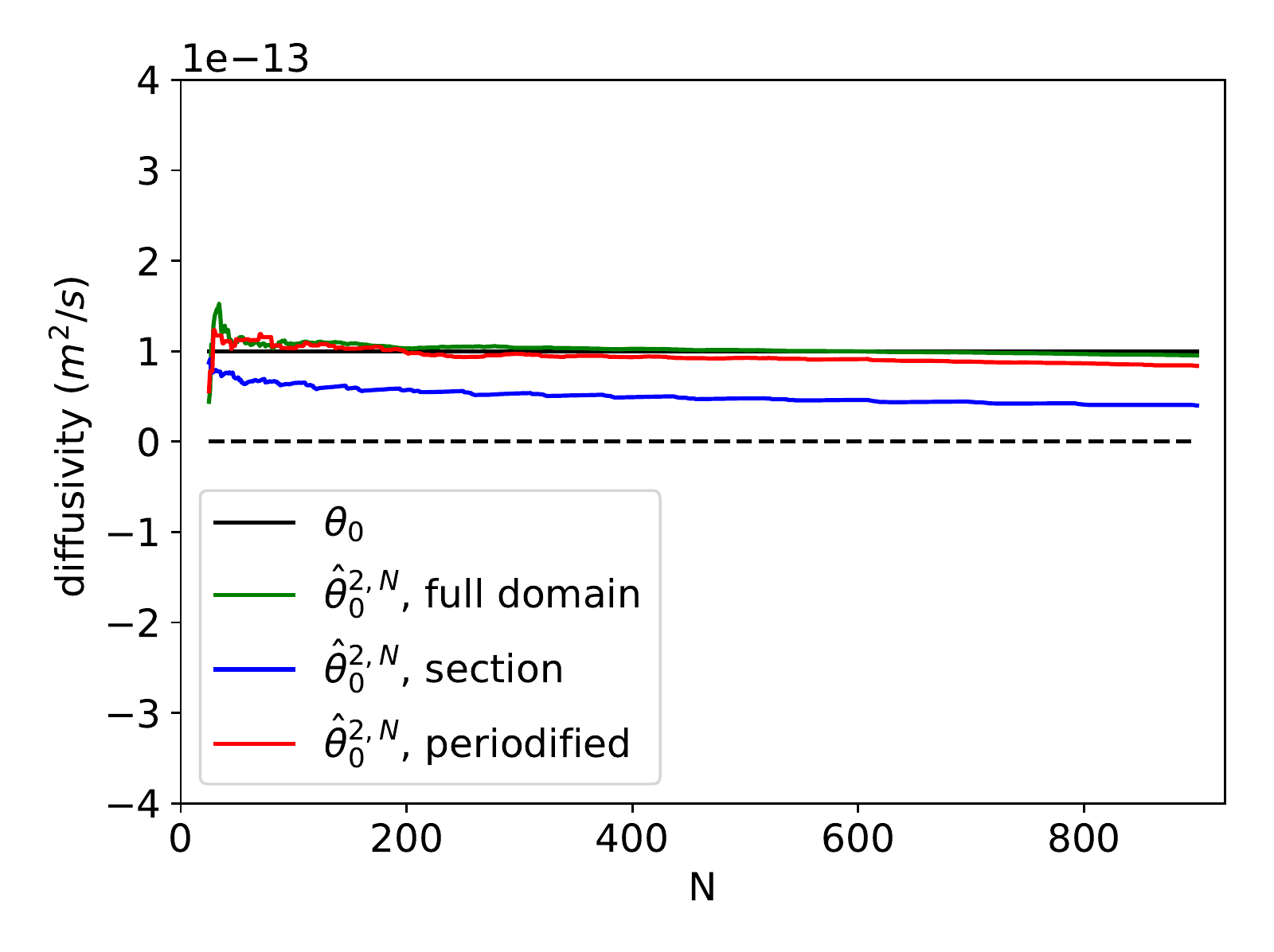}
	\caption{Performance of diffusivity estimators on simulated data under different model assumptions in the spatial high-frequency regime. Solid black line is plotted at the true parameter $\theta_0=1\times10^{-13}$, dashed black line is plotted at zero. In all displays, we restrict to $N\geq 25$ in order to avoid artefacts stemming from low resolution.}
    \label{fig:SimulationOne}
\end{figure}

First, we study the performance of the mentioned estimators on simulations. 
The numerical scheme is specified in Appendix \ref{app:methods}. 
While we have perfect knowledge on the dynamical system which generates the data in this setting, it is revelatory to compare the different versions of $\hat\theta^N_0$ which correspond to varying levels of model misspecification. 
The simulation shows that $a(|U_t|_{L^2})$ 
fluctuates around a value 
slightly larger than 
$0.15$.
We demonstrate the effect of qualitatively different model assumptions on our method in Figure \ref{fig:SimulationOne} (top left) by comparing the performance of $\thetalin$, $\thetapol$ and $\thetafull$. The result can be interpreted as follows: As $\thetalin$ does not see any information on the wave fronts, the steep gradient at the transition phase leads to a low diffusivity estimate. On the other hand, $\thetapol$ incorporates knowledge on the wave fronts as they appear in the Schl\"ogl model, but the decay in concentration due to the presence of the inhibitor is mistaken as additional diffusion. Finally, $\thetafull$ contains sufficient information on the dynamics to give a precise estimate. In Figure \ref{fig:SimulationOne} (top right), we show the effect of wrong a priori assumptions on $\bar a$ in $\thetafull$. Even for $N=800$, the precision of $\thetafull$ clearly depends on the choice of $\bar a$. 
Remember that there is no true $\bar a$ in the underlying model, rather, $\bar a$ serves as an approximation for $a(|U_t|_{L^2})$. 
Better results can be achieved with $\thetatwo, \thetathree$ and $\thetafour$, see Figure \ref{fig:SimulationOne} (bottom left): $\thetatwo$ has no knowledge on $\bar a$ and recovers the diffusivity precisely, and even $\thetafour$ performs better than the misspecified $\thetafull$ from the top right panel of Figure \ref{fig:SimulationOne}. 

\subsection{Discussion of the periodic boundary}\label{sec:PerformancePeriodic}
In Figure \ref{fig:SimulationOne} (bottom right) we sketch how the assumption of periodic boundary conditions influences the estimate. While $\thetatwo$ works very well on the full domain of $200\times 200$ pixels with periodic boundary conditions, it decays rapidly if we just use a square section of $75\times 75$ pixels. In fact, the boundary conditions are not satisfied on that square section. This leads to the presence of discontinuities at the boundary. These discontinuities, if interpreted as steep gradients, lower the observed diffusivity. Hence, a first guess to improve the quality is to mirror the square section along each axis and glue the results together. In this manner we construct a domain with $150\times 150$ pixels, on which $\thetatwo$ performs well. We emphasize that, while this periodification procedure is a natural approach, its performance will depend on the specific situation, because the dynamics at the transition edges will still not obey the true underlying dynamics. Furthermore, by modifying the data set as explained, we change its resolution, and consequently, a different amount of spectral information may be included into $\thetatwo$ for interpretable results. 

\subsection{Effect of the Noise Intensity}\label{sec:SimulationSigma}
\begin{figure}[t]
    \includegraphics[width=0.49\textwidth]{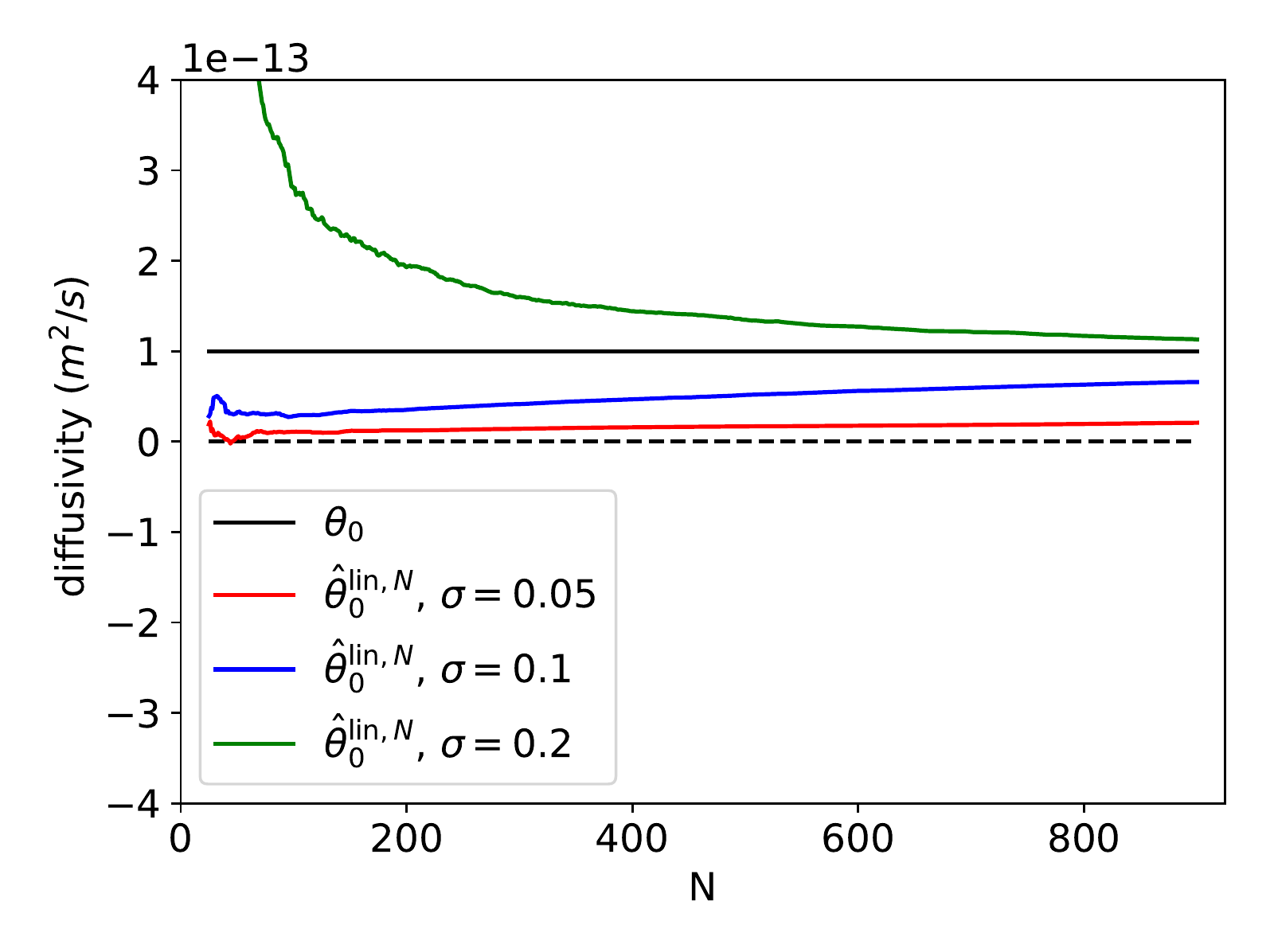}
    \includegraphics[width=0.49\textwidth]{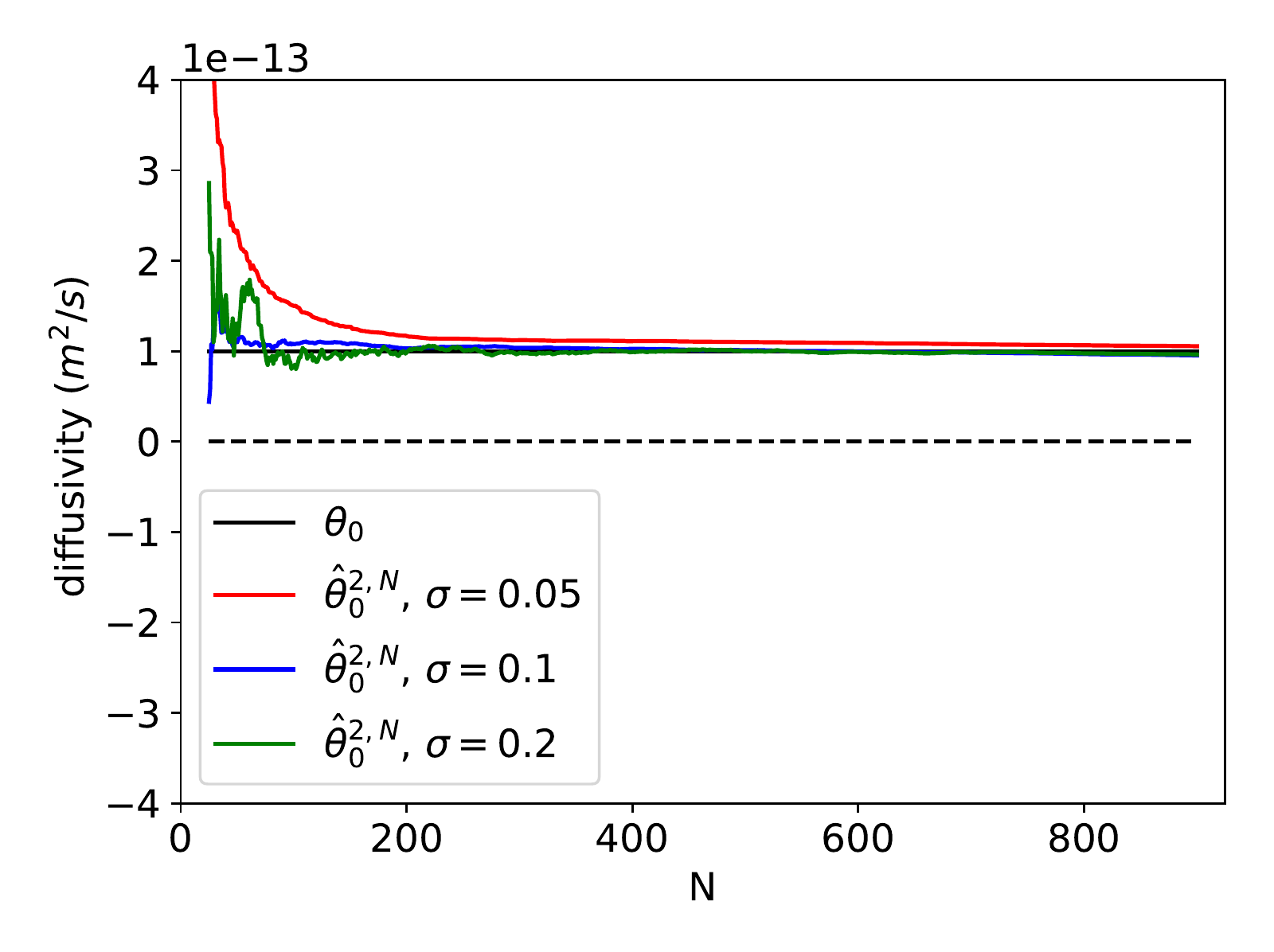}
    \caption{
    Sensitivity of {\bf (left)} $\thetalin$ and {\bf (right)} $\thetatwo$ to different noise levels. Solid black line is plotted at $\theta_0=1\times 10^{-13}$, dashed black line is plotted at zero. As before, we restrict to $N\geq 25$ in the plots.
    \label{fig:SimulationSigma}
    }
\end{figure}
In Figure \ref{fig:SimulationSigma}, we study the effect of varying the noise level in the simulation. We compare $\thetalin$, which is agnostic to the reaction model, to $\thetatwo$, which incorporates a detailed reaction model. While $\thetatwo$ performs well regardless of the noise level, the quality of $\thetalin$ tends to improve for larger $\sigma$. In this sense, a large noise amplitude hides the effect of the nonlinearity. This is in line with the observations made in \cite[Section 3]{PasemannStannat20}. We note that the dynamical features of the process change for $\sigma=0.2$: It is no longer capable of generating travelling waves in that case.

\subsection{Performance on Real Data}\label{sec:RealData}

\begin{figure}[t]
    \includegraphics[width=0.49\textwidth]{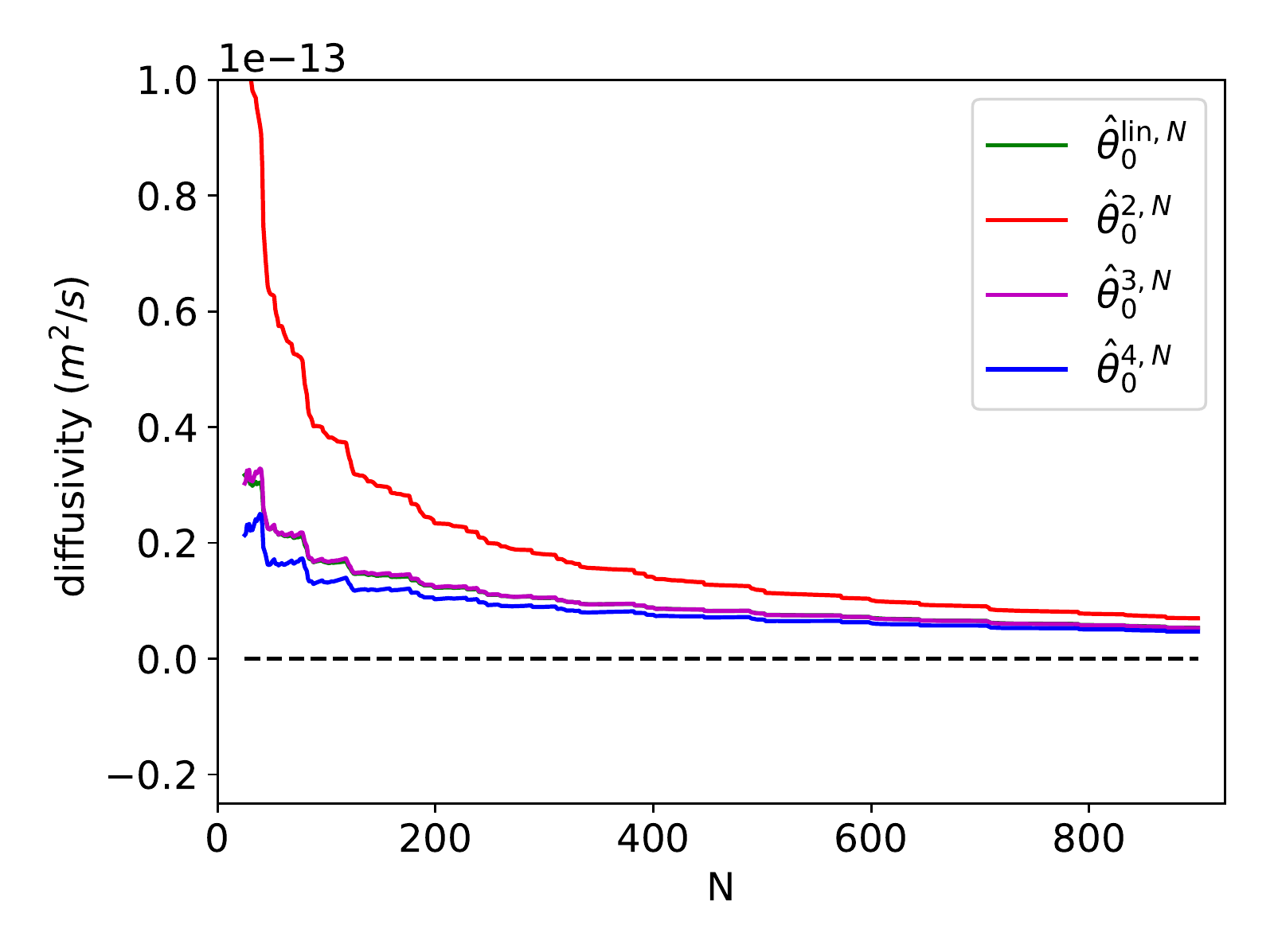}
	\includegraphics[width=0.49\textwidth]{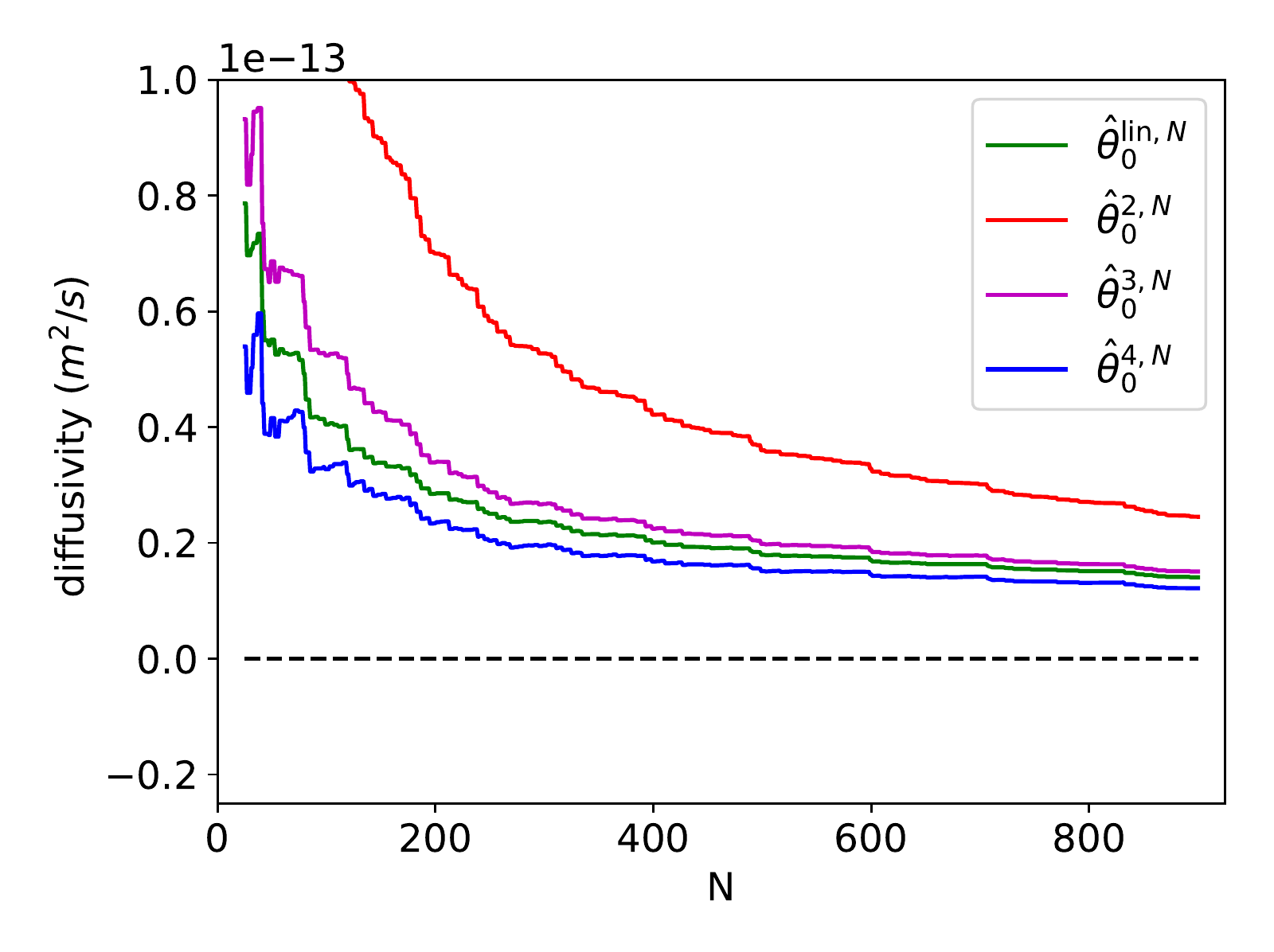}
	\includegraphics[width=0.49\textwidth]{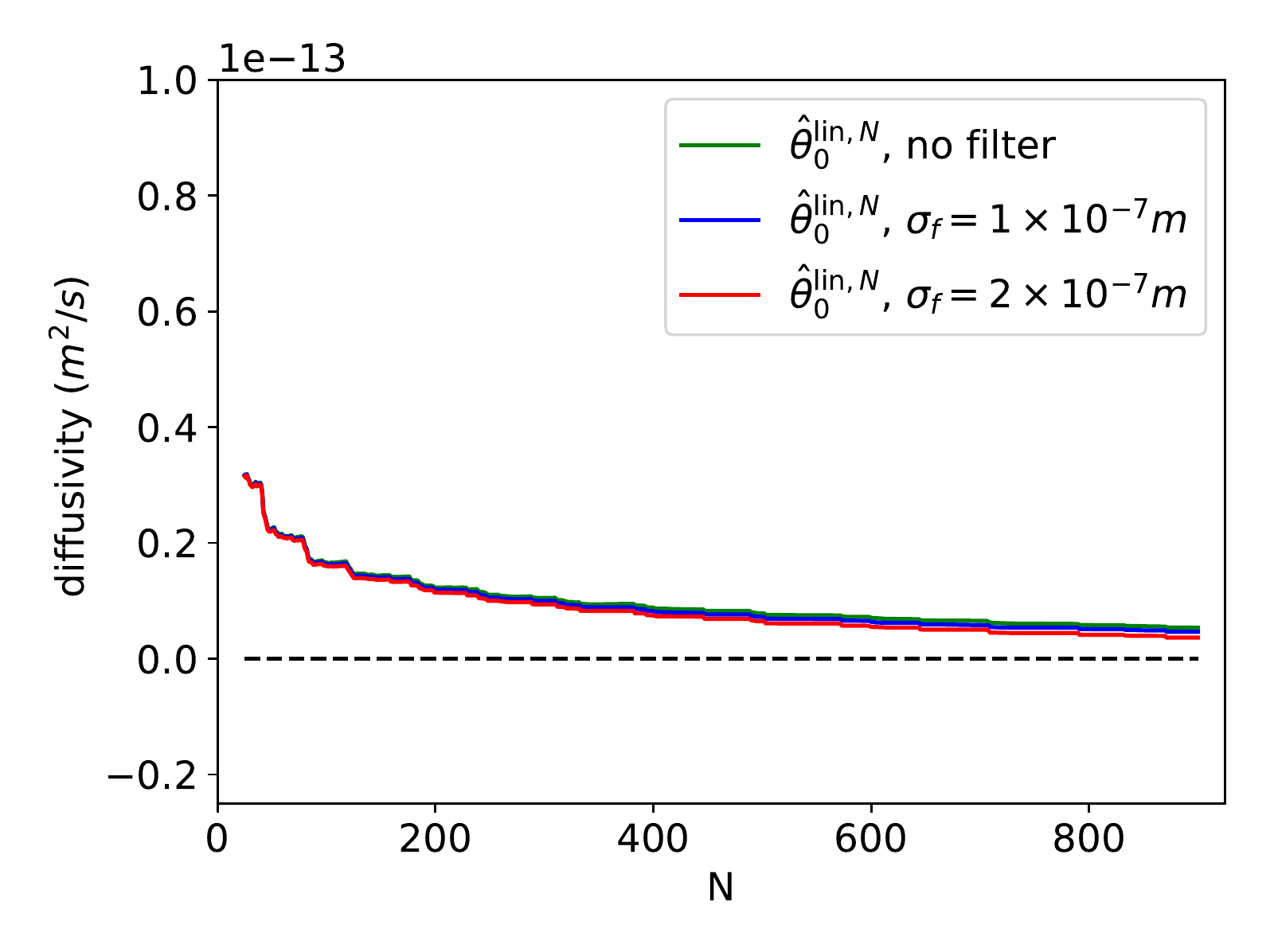}
    \includegraphics[width=0.49\textwidth]{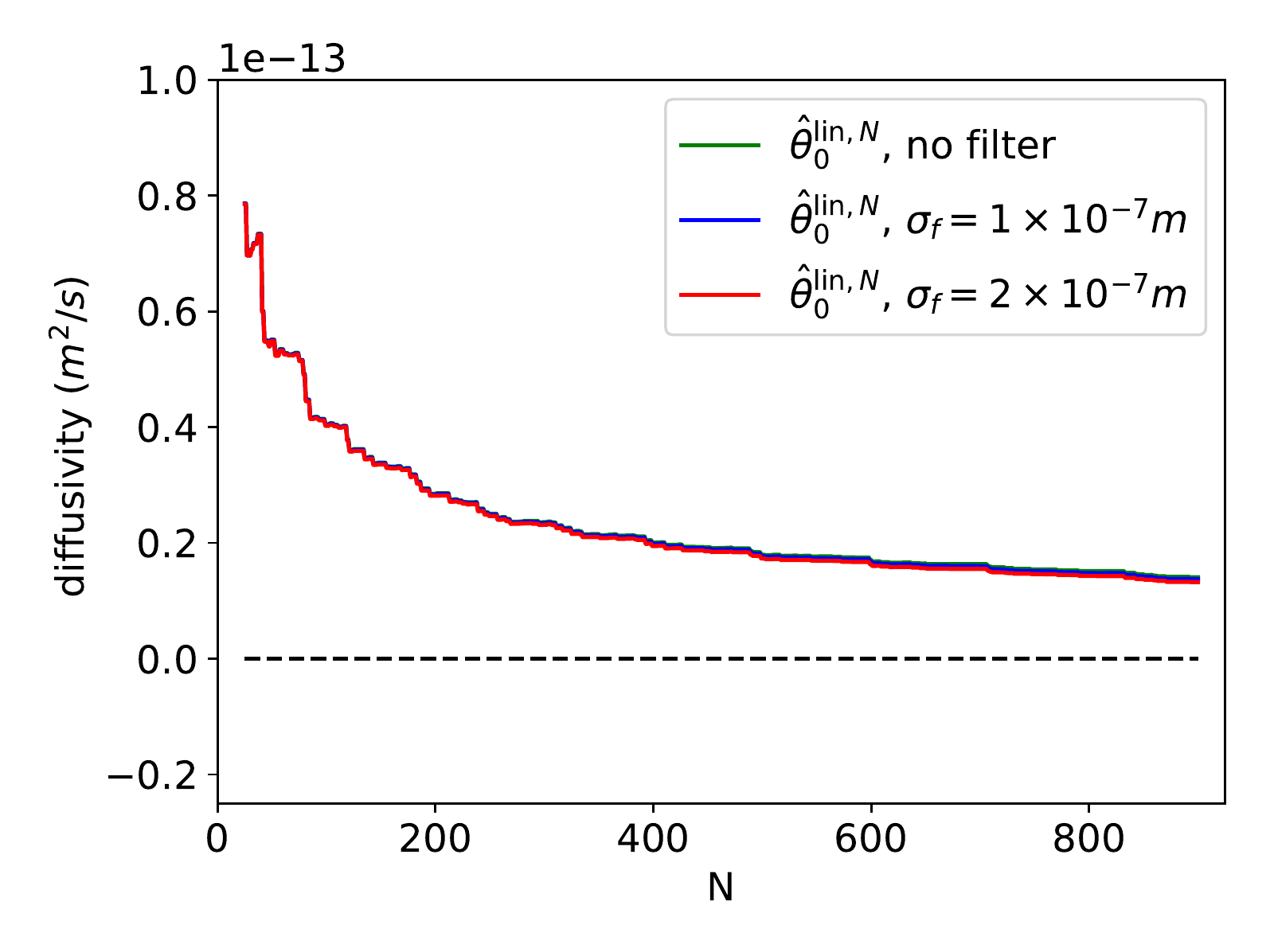}
    \caption{In all displays, we restrict to $N\geq 25$ in order to avoid artefacts stemming from low resolution. Dashed black line is plotted at zero. {\bf (top)} Performance of different diffusivity estimators on (top left) cell data and (top right) periodified cell data. {\bf (bottom)} The effects of applying a kernel with bandwidth $\bar\sigma$ is shown for (bottom left) not periodified and (bottom right) periodified data.
	\label{fig:CelldataOne}
}
\end{figure}

A description of the experimental setup can be found in Appendix \ref{app:methods}. 
The concentration in the data is represented by grey values ranging from $0$ to $255$ at every pixel. 
This range is standardized to the unit interval $[0,1]$, in order to match the stable fixed points of the bistable polynomial $f$ in the reference case $u_0=1$.
Note that this is necessary for all estimators except $\thetalin$. 
We compare $\thetalin$ with $\thetatwo$, $\thetathree$ and $\thetafour$, which are more flexible than $\thetapol$ and $\thetafull$.
In Figure \ref{fig:CelldataOne} (top left) the behaviour of 
these four estimators on a sample cell is shown. 
Interestingly, 
the model-free linear estimator $\thetalin$ is close to $\thetathree$ and $\thetafour$, which impose very specific model assumptions. This pattern can be observed across different cell data sets. In particular, this is notably different from the performance of these three estimators on synthetic data. 
This discrepancy seems to indicate that the lower order reaction terms in the activator-inhibitor model are not fully consistent with the information contained in the experimental data. This can have several reasons, for example, it is possible that a more detailed model reduction of the known signalling pathway inside the cell is needed. On the contrary, $\thetatwo$ seems to be comparatively rigid due to its a priori choices for $\theta_2$ and $\theta_3$, but it eventually approaches the other estimators. 
Variations in the value of $u_0$ have an impact on the results for small $N$ but not on the asymptotic behaviour. 
In Figure \ref{fig:CelldataOne} (top right), the cell from Figure \ref{fig:CelldataOne} (top left) is periodified before evaluating the estimators. As expected from the discussion in Section \ref{sec:PerformancePeriodic}, the estimates rise, but the order of magnitude does not change drastically.

\subsection{Invariance under Convolution}
Given a function $k\in L^1(\mathcal{D})$, define $T_k:H^s(\mathcal{D})\rightarrow H^s(\mathcal{D})$, $s\in\R$, via $u\mapsto k*u=\int_\mathcal{D}k(\cdot-x)u(x)\mathrm{d}x$, where $k$ and $u$ are identified with their periodic continuation. It is well-known that $T_k$ commutes with $\Delta$, i.e. $T_k\circ\Delta=\Delta\circ T_k$. Thus, if $X$ is a solution to a semilinear stochastic PDE with diffusivity $\theta$, the same is true for $T_kX$: While the nonlinearity and the dispersion operator may be changed by $T_k$, the diffusive part is left invariant, in particular the diffusivity of $X$ and $T_kX$ is the same. Based on this observation, a comparison between the effective diffusivity of $X$ and $T_kX$ for different choices of $k$ may serve as an indicator if the assumption that a data set $X$ is generated by a semilinear SPDE \eqref{eq:SPDEmodelGeneral} is reasonable in the first place, and if the diffusion indeed can be considered to be homogeneous and anisotropic. 
We use a family of periodic kernels $k=k_{\bar \sigma}$, $\bar\sigma>0$, which are normed in $L^1(\mathcal{D})$ and coincide on the reference rectangle $[-L_1/2,L_1/2]\times[-L_2/2,L_2/2]$ with a Gaussian density with standard deviation $\bar\sigma$. In Figure \ref{fig:CelldataOne} (bottom), the effects of applying $T_{k_{\bar\sigma}}$ for different bandwidths $\bar\sigma$ are shown for one data set and its periodification. While the diffusivity of the data without periodification on the left-hand panel is slightly affected by the kernel, more precisely, its tendency to fall is enlarged, 
the graphs for the effective diffusivity of the periodified data are virtually indistinguishable. 
Periodification seems to be compatible with the expected invariance under convolution, 
even if the periodified data is not generated by a semilinear SPDE but instead by joining smaller patches of that form. In total, these observations are in accordance with the previous sections and suggest that the statistical analysis of the data based on a semilinear SPDE model is reasonable. 

\subsection{The Effective Diffusivity of a Cell Population}\label{sec:CellPopulation}

\begin{figure}[t]
    \includegraphics[width=0.49\textwidth]{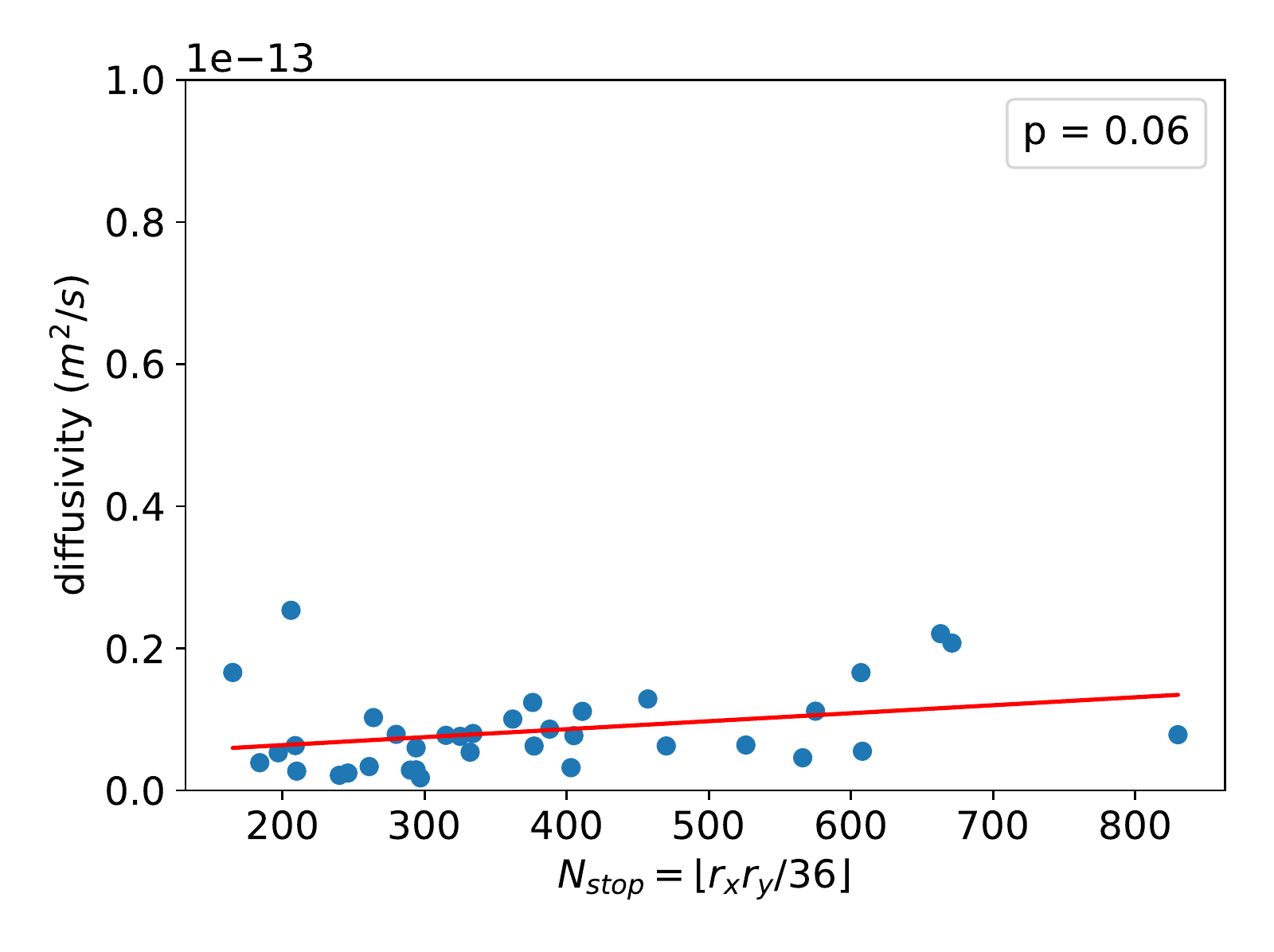}
    \includegraphics[width=0.49\textwidth]{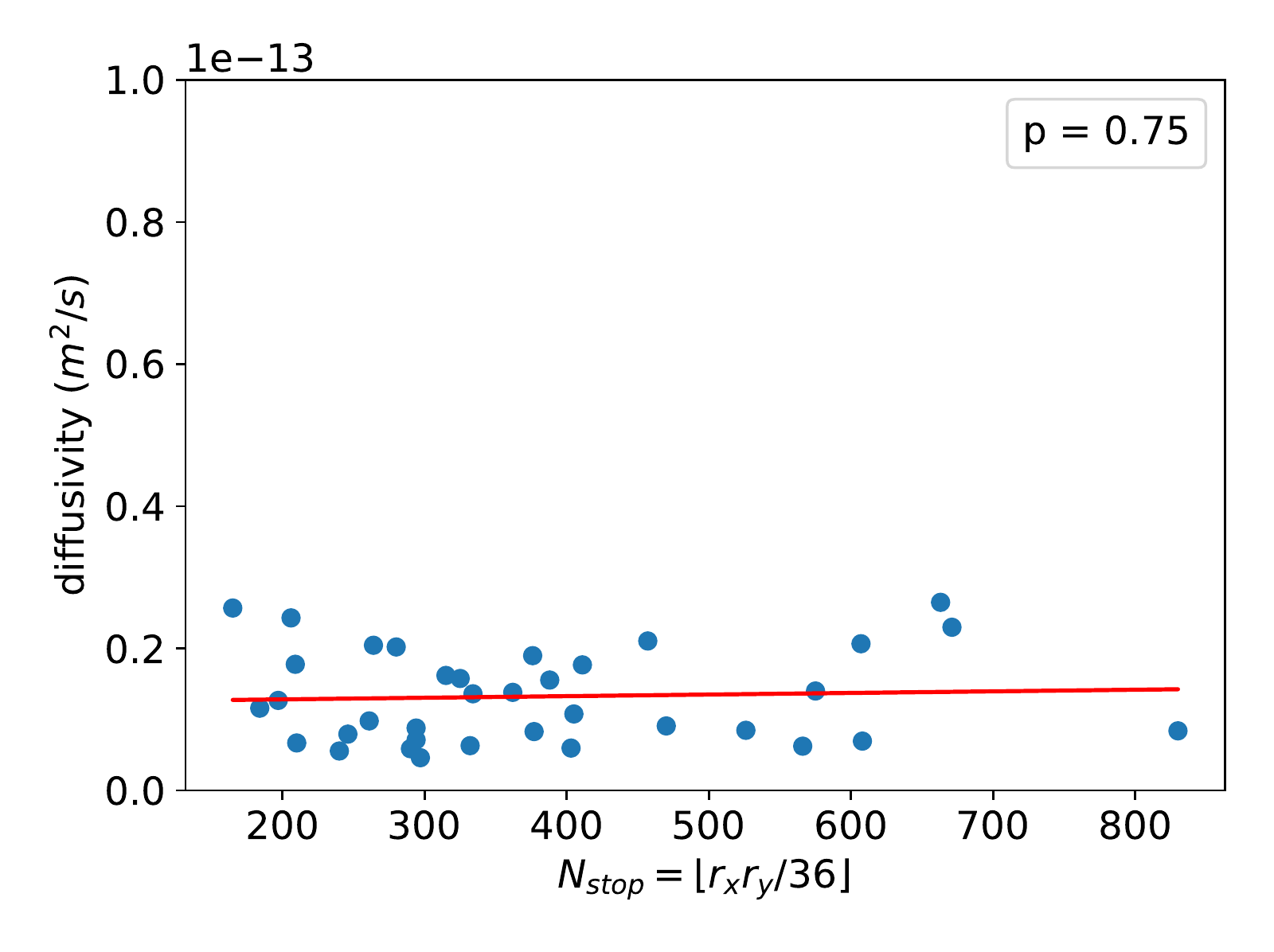}
    \includegraphics[width=0.49\textwidth]{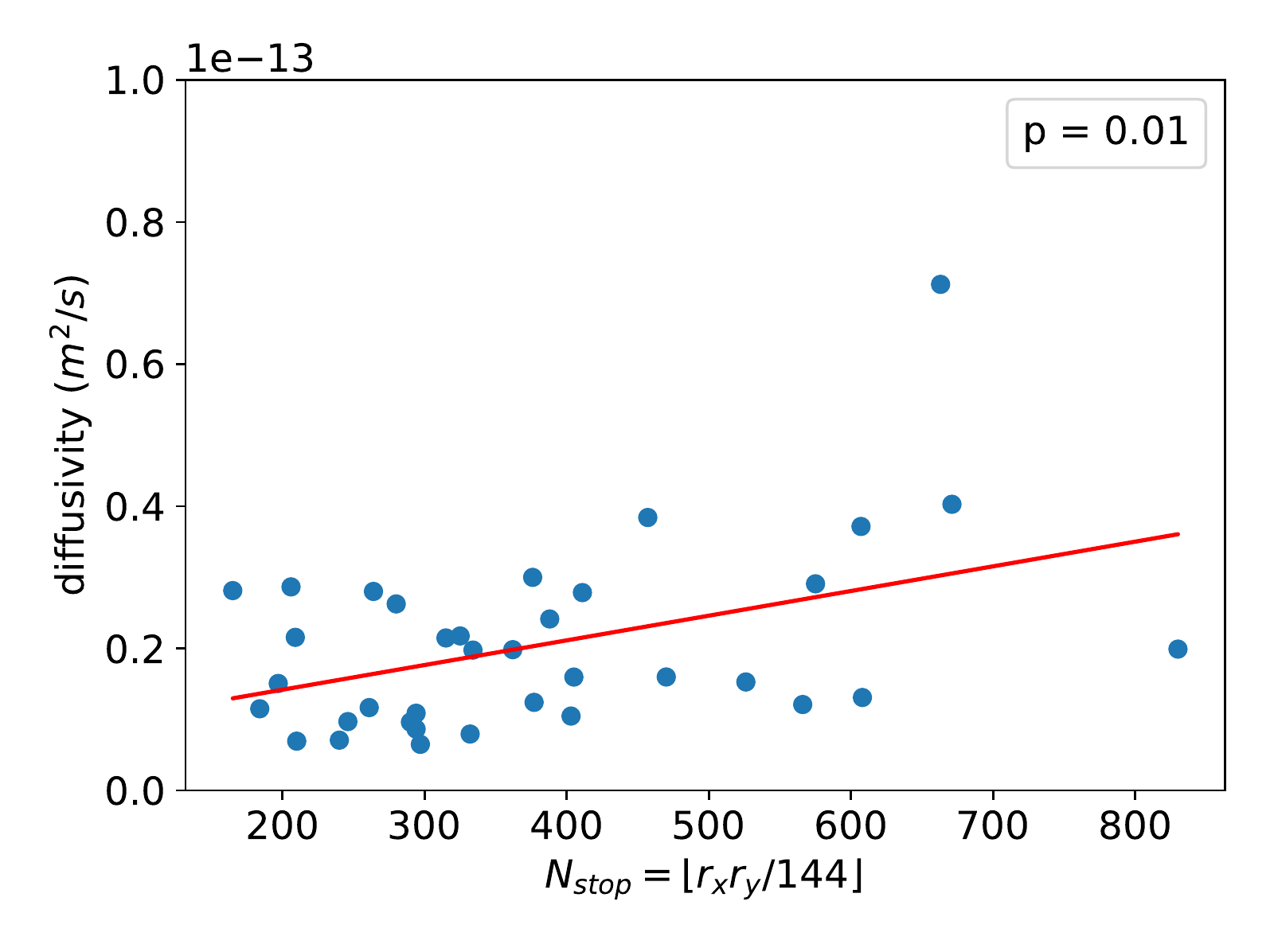}
    \includegraphics[width=0.49\textwidth]{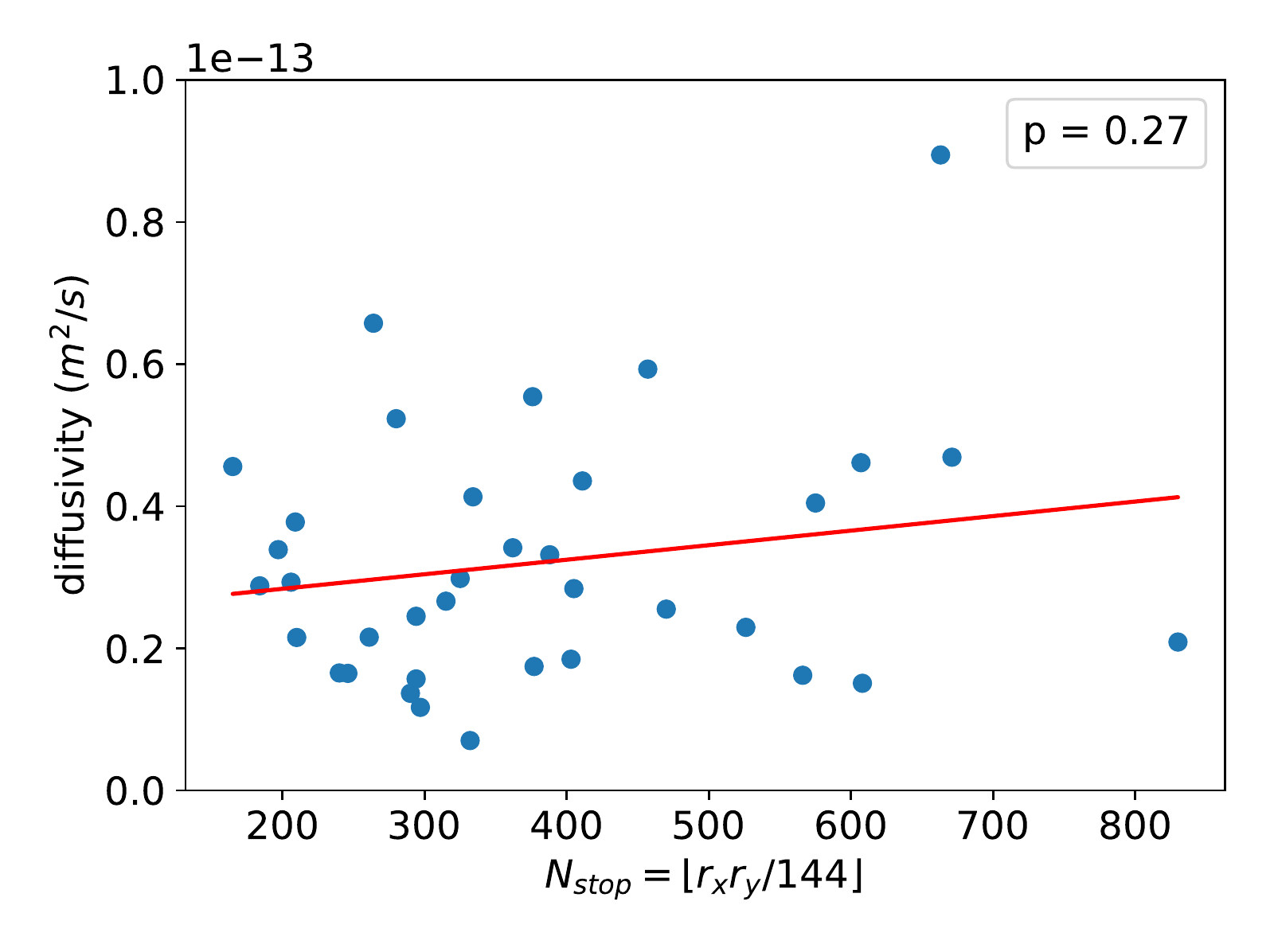}
    \caption{The samples are evaluated {\bf (top)} without or {\bf (bottom)} with periodification. 
    $\hat\theta^{3, N}_0$ 
    with {\bf (left)} $N=N_\mathrm{const}$ or {\bf (right)} $N=N_\mathrm{stop}$ 
    is plotted against $N_\mathrm{stop}$. 
    The least squares fit is shown in red. The $p$-value in each plot corresponds to a t-test whose null hypothesis states that the slope of the regression line is zero. Clearly, the slope is more notable in the case $N=N_{\mathrm{const}}$.
    }
    \label{fig:CelldataPopulation}
\end{figure}

We compare the estimated diffusivity for a cell population consisting of $36$ cells. The boundaries in space and time of all samples are selected in order to capture only the interior dynamics within a cell and, consequently, the data sets differ in their 
size. 
On the one hand, the estimated diffusivity tends to stabilize in time, i.e. the number of frames in a sample, corresponding to the final time $T$, does not affect the result much. On the other hand, the size of each frame, measured in pixels, determines the number of eigenfrequencies that carry meaningful information on the process. Thus, we expect that $N$ can be chosen larger for samples with high spatial resolution. We formalize this intuition with the following heuristic: Let $r_x$ and $r_y$ denote the number of pixels in each row and column, resp., of every frame in a sample. Let $N_\mathrm{stop}=\lfloor 4r_xr_y/M^2\rfloor$, where $M\in\N$ is a parameter representing the number of pixels needed for a sine or cosine to extract meaningful information. That is, a square of dimensions $r_x\times r_y = M\times M$ leads to $N_\mathrm{stop}=4$, so in this case only the eigenfunctions $\Phi_{k,l}$, $k,l\in\{-1,1\}$ are taken into account, whose period length is $M$ in both dimensions. In our evaluations, we choose $M=12$ for the data without periodification and $M=24$ for the periodified data sets. We mention that the cells are also heterogeneous with respect to their characteristic length $\overline{dx}$ and time $\overline{dt}$, given in meters per pixel and seconds per frame for each data set.\footnote{Further tests indicate that the estimated diffusivity correlates with the characteristic diffusivity $\overline{dx}^2/\overline{dt}$. A more detailed examination of the discretization effects and other dependencies which are not directly related to the underlying process is left for future research. 
Also, it would be interesting to understand the extent to which the effective diffusivity in the data is scale dependent.
} However, a detailed quantitative analysis of the resulting discretization error is beyond the scope of our work. 
In Figure \ref{fig:CelldataPopulation}, we compare the results for $\hat\theta^{3,N_\mathrm{stop}}_0$ with $\hat\theta^{3,N_\mathrm{const}}_0$ within the cell population, where $N_\mathrm{const}=899$ is independent of the sample resolution. Evaluating the estimator at $N_\mathrm{stop}$ decorrelates the estimated diffusivity from the spatial extension of the sample. The results for different cells have the same order of magnitude, which indicates that the effective diffusivity can be used for statistical inference for cells within a population or between populations in future research.

\subsection{Estimating $\theta_1$.} 

\begin{figure}[t]
	\includegraphics[width=0.49\textwidth]{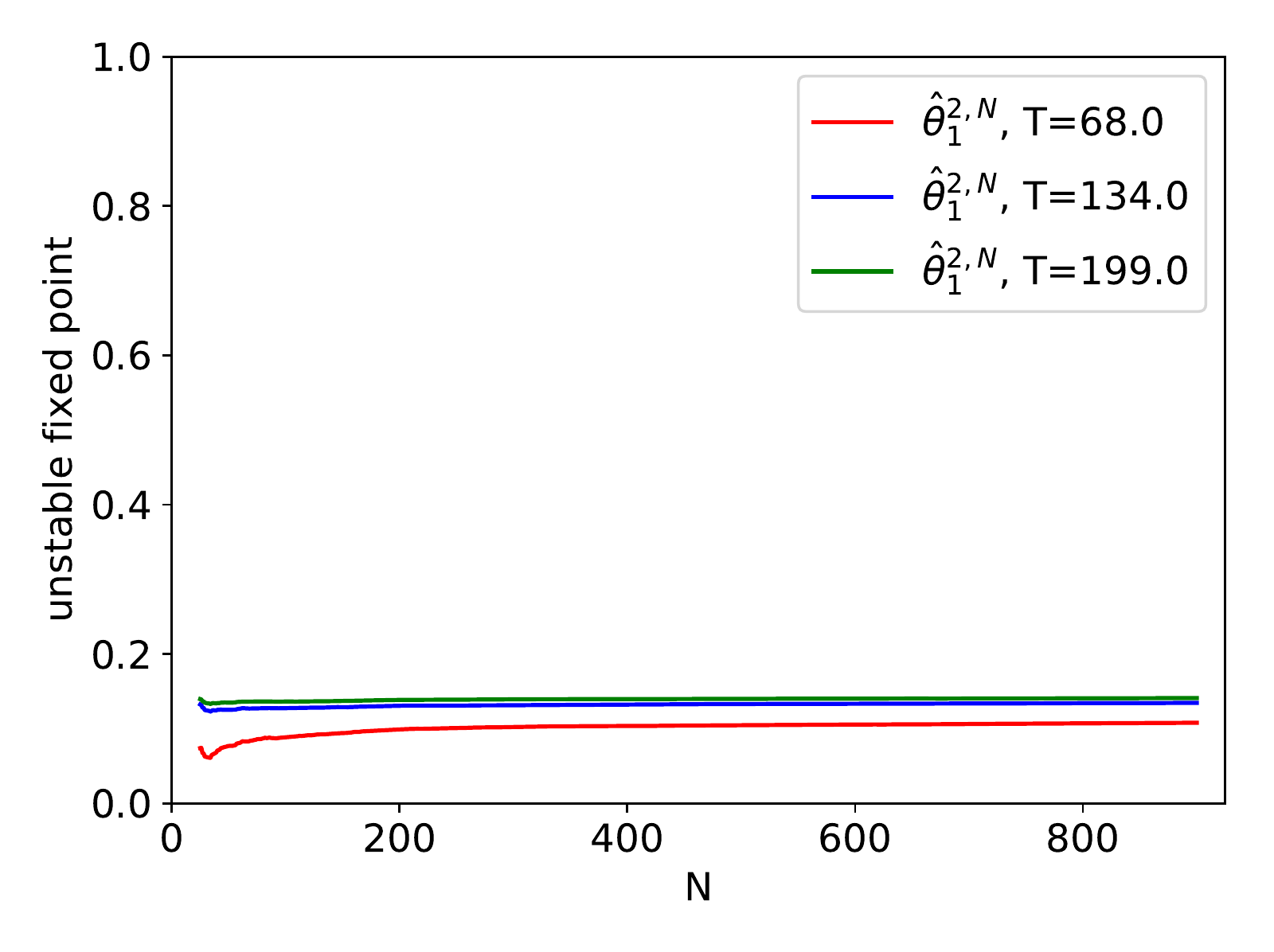}
	\includegraphics[width=0.49\textwidth]{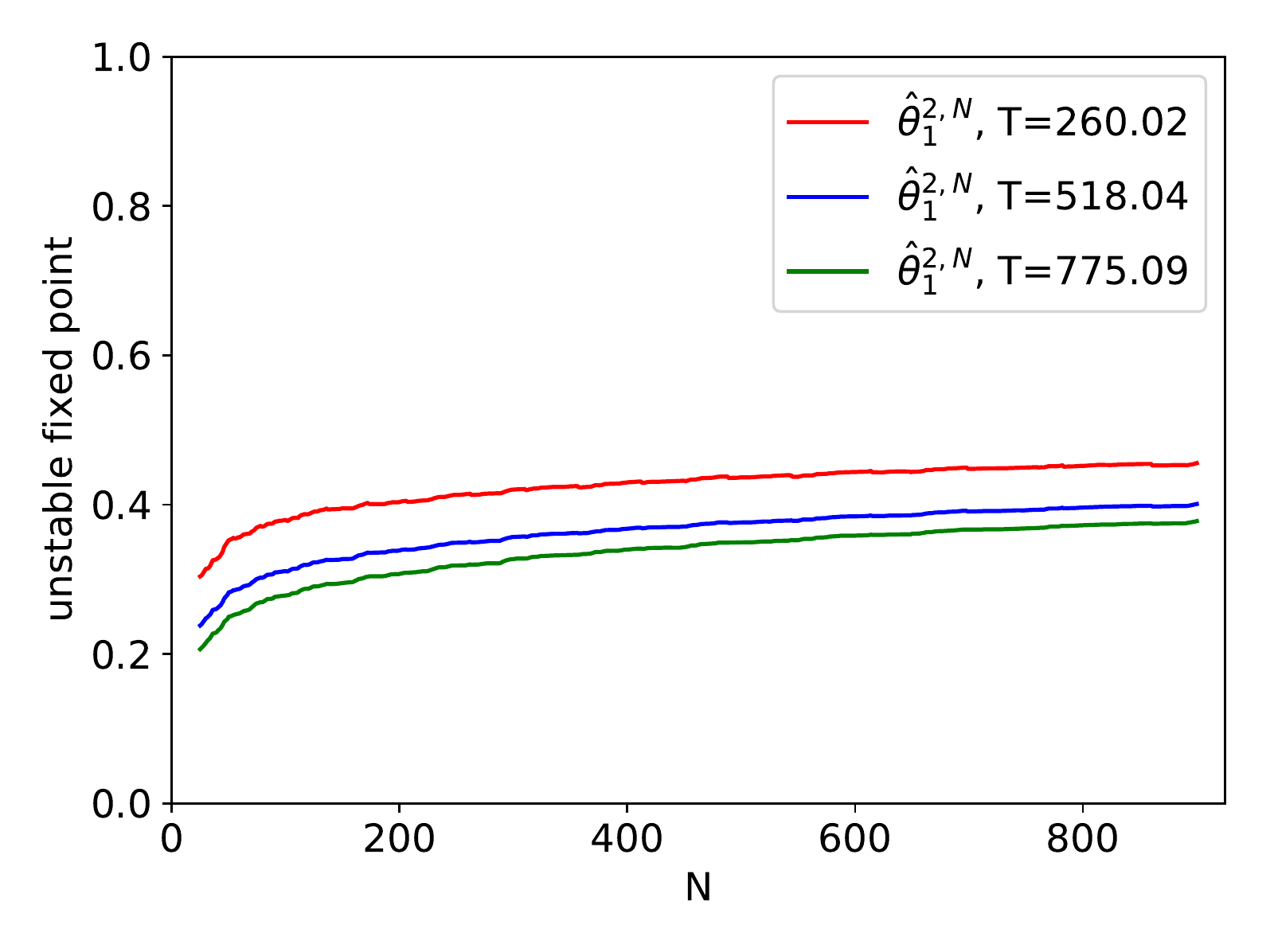}
	\caption{Estimation of the reaction parameter $\theta_1$ on {\bf (left)} simulated data and {\bf (right)} cell data. All times are given in seconds, relative to the first frame. As before, we restrict to $N\geq 25$.
	\label{fig:EstimateThetaOne}
}
\end{figure}

When solving the linear MLE equations in order to obtain $\thetatwo$, we simultaneously get an estimate $\hat\theta_1^{2,N}$ for $\theta_1=k_1u_0\bar a$. Note that $u_0=1$ by convention, and $\theta_2=k_1=1$ is treated as known quantity in this case. Thus, $\theta_1$ can be identified with $\bar a$. 
In Figure~\ref{fig:EstimateThetaOne} we show the results for $\hat\theta_1^{1,N}$ on simulated data (left) and on a cell data sample (right). 
Note that in general, we cannot expect 
to observe increased precision for $\hat\theta_1^{2,N}$ as $N$ grows, because the reaction term is of order zero.\footnote{According to \cite{HuebnerRozovskii95}, the maximum likelihood estimator for the coefficient of a linear order zero perturbation to a heat equation with known diffusivity converges only with logarithmic rate in $d=2$.}
However, it is informative to consider also the large time regime $T\rightarrow\infty$, i.e. to include more and more frames to the estimation procedure. 
In the case of simulated data, $a$ oscillates around an average value 
close to $0.15$, 
which should be considered to be the ground truth for $\bar a$. This effective value $\bar a$ is recovered well, 
even for small $T$, with increasing precision as $T$ grows. 
Clearly, this depends heavily on the model assumptions. In the case of cell data, the results are rather stable. 
This indicates that it may be reasonable to use the concept of an ``effective unstable fixed point $\bar a$ of the reaction dynamics, conditioned on the model assumptions included in $\thetatwo$'', when evaluating cell data statistically.

\subsection{The Case of Pure Noise Outside the Cell}

\begin{figure}
	\includegraphics[width=0.49\textwidth]{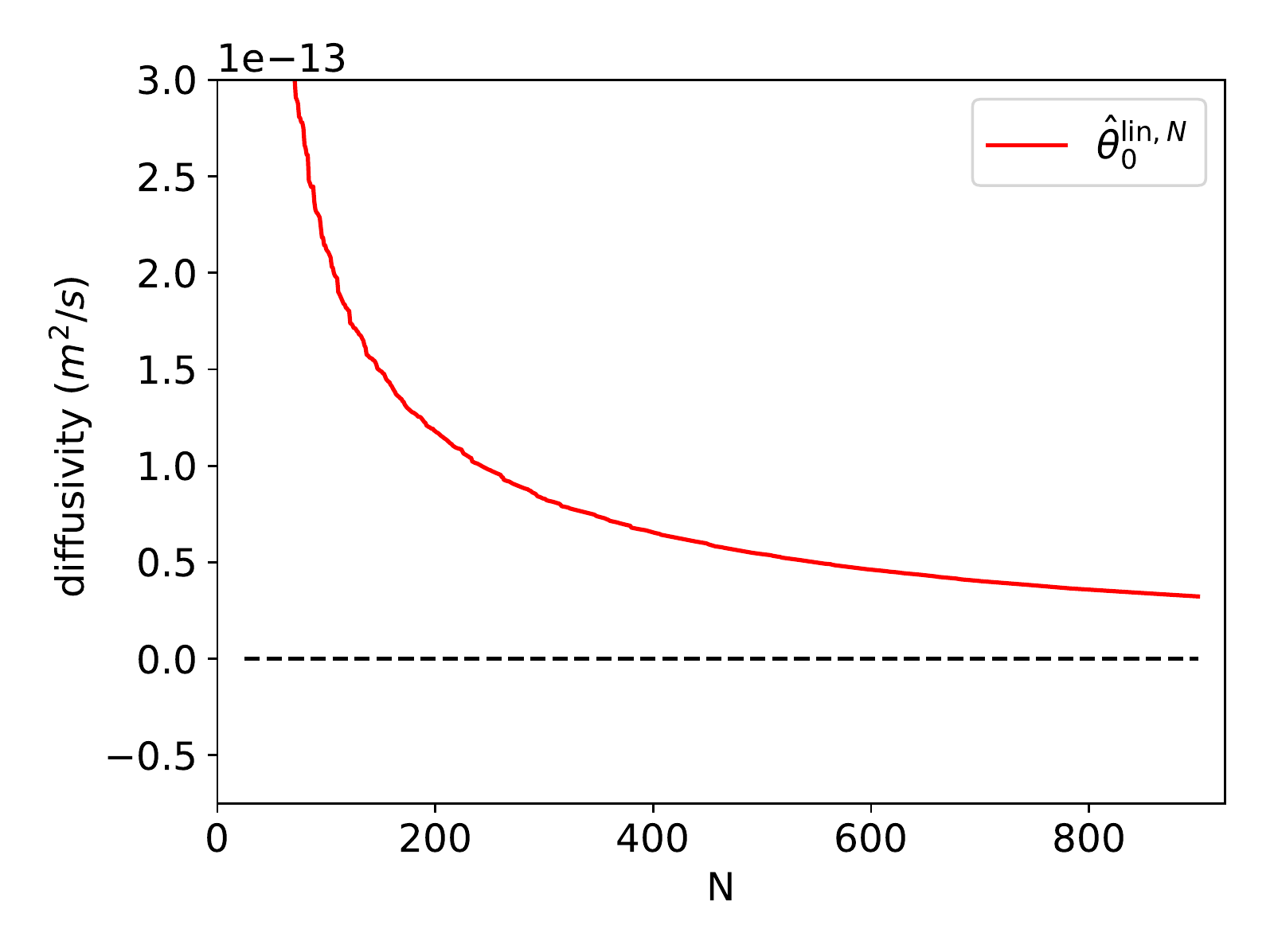}
	\includegraphics[width=0.49\textwidth]{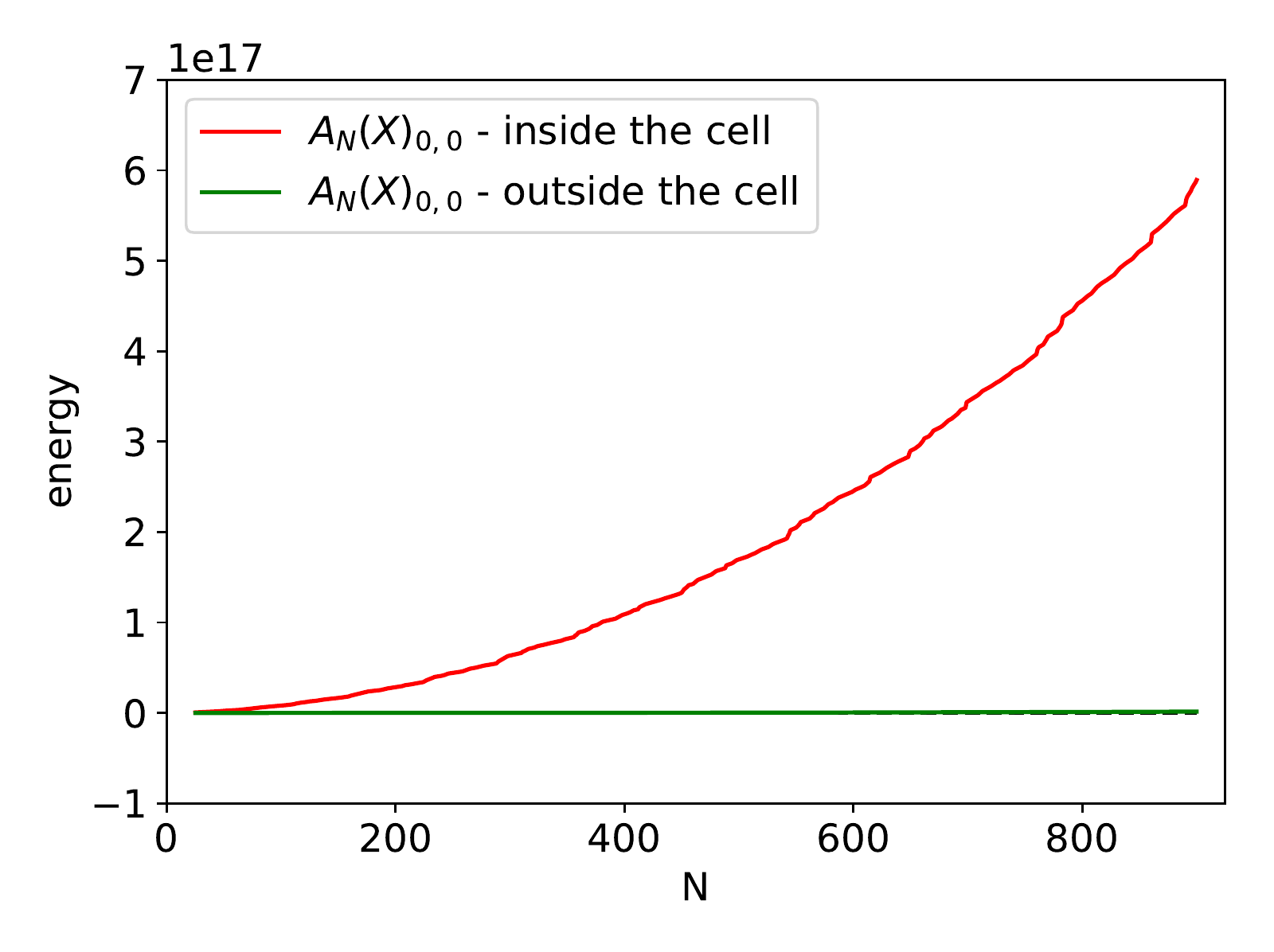}
	\caption{As before, we restrict to $N\geq 25$. {\bf (left)} Effective diffusivity outside the cell, plot for one data set. Dashed line is plotted at zero. {\bf (right)} Comparison of the energy inside and outside the cell. Both data sets have the same spatial and temporal extensions.}
	\label{fig:CelldataTwo}
\end{figure}

If the data set does not contain parts of the cell but rather mere noise, the estimation procedure still returns a value. This ``observed diffusivity'' (see Figure \ref{fig:CelldataTwo}) originates completely from white measurement noise. More precisely, the appearance and vanishing of singular pixels is interpreted as instantaneous (i.e. within the time between two frames) diffusion to the steady state. Thus, the observed diffusivity in this case can be expected to be even larger than the diffusivity inside the cell. In this section, we give a heuristical explanation for the order of magnitude of the effective diffusivity outside the cell. \\

We work in dimension $d=2$. Assume that a pixel has width $x>0$. This value is determined by the spatial resolution of the data. For simplicity, we approximate it by a Gaussian density $\phi_0(y)$ with standard deviation $\sigma_0 = \frac{x}{2}$. This way, the inflection points of the (one-dimensional marginal) density match the sharp edges of the pixel. Now, using $\phi_0$ as an initial condition for the heat equation on the whole space $\R^2$, the density $\phi_t$ after time $t$ is also a Gaussian density, with standard deviation $\sigma_t=\sqrt{\sigma_0^2+2\theta t}$, obtained by convolution with the heat kernel. The maximal value $f_\mathrm{max}^t$ of $\phi_t$ is attained at $y=0$ with $f_\mathrm{max}^t = (2\pi\sigma_t^2)^{-1} = (2\pi(\sigma_0^2 + 2\theta t))^{-1}$. Now, if we observe after time $t>0$ at the given pixel an intensity decay by a factor $b>0$, i.e.
\begin{equation}
	bf_\mathrm{max}^t\leq f_\mathrm{max}^0,
\end{equation}
this leads to an estimate for the diffusivity of the form
\begin{equation}
	\theta\geq (b-1)\frac{\sigma_0^2}{2t}. 
\end{equation}
For example, set $t = 0.97s$ and $x=2.08\times 10^{-7}m$, as in the data set from Figure \ref{fig:CelldataTwo} (left). The intensity decay factor varies between different pixels in the data set, reasonable values are given for $b\leq 30$. If $b=30$, we get $\theta\geq 1.6\times 10^{-13}m^2/s$, for $b = 20$, we get $\theta\geq 1\times 10^{-13}m^2/s$, and for $b=15$, we have $\theta\geq 7.8\times 10^{-14}m^2/s$. This matches the observed diffusivity outside the cell from Figure \ref{fig:CelldataTwo}, which is indeed of order $1\times 10^{-13}m^2/s$: For example, with $N=N_\mathrm{stop}=\lfloor 4r_xr_y/M^2\rfloor$ and $M=12$, as in Section \ref{sec:CellPopulation}, we get $N_\mathrm{stop}=165$ and $\hat\theta^{\mathrm{lin}, N_\mathrm{stop}}_0=1.36\times 10^{-13}m^2/s$ for this data set consisting of pure noise. In total, this gives a heuristical explanation for the larger effective diffusivity outside the cell compared to the estimated values inside the cell. \\

It is important to note that even if the effective diffusivity outside the cell is larger, this has almost no effect on the estimation procedure inside the cell. This is because the total energy $A_N(X)_{0, 0}$ of the noise outside the cell is several orders of magnitude smaller than the total energy of the signal inside the cell, see Figure \ref{fig:CelldataTwo} (right). 

\section{Discussion and Further Research}

In this paper we have extended the mathematical theory of parameter estimation of stochastic reaction-diffusion system to the joint estimation problem of diffusivity and  parametrized reaction terms within the variational theory of stochastic partial differential equations. We have in particular applied our theory to the estimation of effective diffusivity of intracellular actin cytoskeleton dynamics. 

Traditionally, biochemical signaling pathways were studied in a purely temporal manner, focusing on the reaction kinetics of the individual components and the sequential order of the pathway, possibly including feedback loops.
Relying on well-established biochemical methods, many of these temporal interaction networks could be characterized.
However, with the recent progress in the {\it in vivo} expression of fluorescent probes and the development of advanced live cell imaging techniques, the research focus has increasingly shifted to studying the full spatiotemporal dynamics of signaling processes at the subcellular scale.
To complement these experiments with modeling studies, stochastic reaction-diffusion systems are the natural candidate class of models that incorporate the relevant degrees of freedom of intracellular signaling processes. 
Many variants of this reaction-diffusion framework have been proposed in an empirical manner to account for the rich plethora of spatiotemporal signaling patterns that are observed in cells.
However, the model parameters in such studies are oftentimes chosen in an {\it ad hoc} fashion and tuned based on visual inspection, so that the patterns produced in model simulations agree with the experimental observations.
A rigorous framework that allows to estimate the parameters of stochastic reaction-diffusion systems from experimental data will provide an indispensable basis to refine existing models, to test how well they perform, and to eventually establish a new generation of more quantitative mathematical models of intracellular signaling patterns.

The question of robustness of the parameter estimation problem w.r.t. specific modelling assumptions of the underlying stochastic evolution equation is an important problem in applications that needs to be further investigated in future research. 
In particular, this applies to the dependence of diffusivity estimation on the domain and its boundary. In this work, we based our analysis on a Fourier decomposition on a rectangular domain with periodic boundary conditions. A natural, boundary-free approach is using 
local estimation techniques 
as they have been developed and used in \cite{AltmeyerReiss19, AltmeyerCialencoPasemann20, AltmeyerBretschneiderJanakReiss20}. 
An additional approach aiming in the same direction is the application of a wavelet transform.

It is a crucial task to gather further information on the reaction term from the data. Principally, this cannot be achieved in a satisfactory way on a finite time horizon, so the long-time behaviour of maximum likelihood-based estimators needs to be studied in the context of stochastic reaction-diffusion systems. We will address this issue in detail in future work. 

\appendix

\section{Additional Proofs}

\subsection{Proof of Proposition \ref{prop:FNWellPosed}}\label{app:FNWellPosed}

We prove Proposition \ref{prop:FNWellPosed} by a series of lemmas. First, we show that \eqref{eq:Activator}, \eqref{eq:Inhibitor} is well-posed in $\mathcal H = L^2\times L^2$. First note that for any $u,v\in L^2(\mathcal{D})$ and $x\in\R$:
\begin{align}\label{eq:fBoundedAbove}
	\langle \partial_uf(x,u)v,v\rangle_{L^2} \lesssim |v|_{L^2}^2
\end{align}
as $\partial f(x, u)$ is bounded from above uniformly in $x$ and $u$.

\begin{lemma}\label{lem:FNWellPosed}
	There is a unique $\mathcal{H}$-valued solution $(U,V)$ to \eqref{eq:Activator}, \eqref{eq:Inhibitor}, and 
	\begin{align}
		\mathbb{E}\left[\sup_{0\leq t\leq T}|(U_t, V_t)|_{\mathcal{H}}^p\right] < \infty
	\end{align}
	for any $p\geq 1$. 
\end{lemma}
In particular,
\begin{align}\label{eq:FNBasicRegularity}
	\mathbb{E}\left[\sup_{0\leq t\leq T}|U_t|_{L^2}^p\right] < \infty, \quad \mathbb{E}\left[\sup_{0\leq t\leq T}|V_t|_{L^2}^p\right] < \infty
\end{align}
for any $p\geq 1$. 

\begin{proof}[Proof of Lemma \ref{lem:FNWellPosed}]
	This follows from \cite[Theorem 5.1.3]{LiuRockner15}. In order to apply this result, we have to test the conditions $(H1)$, $(H2')$, $(H3)$, $(H4')$ (i.e. hemicontinuity, local monotonicity, coercivity and growth) therein. Set $\mathcal V = H^1\times H^1$. Define $A(U, V) = (A_1(U, V), A_2(U, V))$ with
\begin{align}
	A_1(U, V) &= D_U\Delta U + k_1f(\lvert U\rvert_{L^2}, U) - k_2V, \\
	A_2(U, V) &= D_V\Delta V + \epsilon(bU - V).
\end{align}
As $B=(-\Delta)^{-\gamma}$ is constant and of Hilbert-Schmidt type, we can neglect it in the estimates. 
In order to check the conditions , we have to look separately at $A_1$ and $A_2$. The statements for $A_2$ is trivial by linearity, so we test the parts of the conditions corresponding to $A_1$. $(H1)$ is clear as $f(x,u)$ is a polynomial in $u$ and continuous in $x$. 
$(H2)$ follows from \eqref{eq:fBoundedAbove} via
\begin{align*}
	{}_{H^{-1}}\langle A_1(U_1, V_1) - A_1(U_2, V_2), U_1 - U_2\rangle_{H^1} \hspace{-5cm} & \\
		&\lesssim k_1{}_{H^{-1}}\langle f(\lvert U_1\rvert_{L^2}, U_1) - f(\lvert U_2\rvert_{L^2}, U_2),U_1-U_2\rangle_{H^1} \\
		&\quad + k_2|V_1-V_2|_{L^2}|U_1-U_2|_{L^2} \\
		&\lesssim k_1{}_{H^{-1}}\langle \partial_uf(\lvert U_1\rvert_{L^2}\widetilde V)(U_1-U_2),U_1-U_2\rangle_{H^1} \\
		&\quad + k_1{}_{H^{-1}}\langle \partial_x f(\tilde x, U_2)(\lvert U_1\rvert_{L^2}-\lvert U_2\rvert_{L^2}),U_1-U_2\rangle_{H^1} \\
		&\quad + C\lVert(U_1,V_1) - (U_2,V_2)\rVert_{\mathcal{H}}^2 \\
		&\lesssim k_1|\partial_x f(\tilde x, U_2)|_{L^2}\left|\lVert U_1\rVert_{L^2}-\lVert U_2\rVert_{L^2}\right||U_1-U_2|_{L^2} \\
		&\quad + C\lVert(U_1,V_1) - (U_2,V_2)\rVert_{\mathcal{H}}^2 \\
		&\lesssim (1 + |\partial_x f(\tilde x, U_2)|_{L^2})\lVert(U_1,V_1) - (U_2,V_2)\rVert_{\mathcal{H}}^2
\end{align*}
for some $\tilde x\in\R$ and $\widetilde V:\mathcal{D}\rightarrow\R$, 
and $|\partial_x f(\tilde x, U_2)|_{L^2} = |u_0a'(\tilde x)U_2(u_0-U_2)|_{L^2}\lesssim |U_2|_{L^2}+|U_2^2|_{L^2}$. With $|U_2^2|_{L^2} = |U_2|_{L^4}^2\lesssim |U_2|_{H^1}^2$, we see that local monotonicity as in $(H2)$ is satisfied by taking $\rho((u,v)) = c|u|_{H^1}^2\leq c|(u,v)|_{\mathcal{V}}^2$. Now, for $(H3)$,
\begin{align*}
	{}_{H^{-1}}\langle A_1(U,V),U\rangle_{H^1} &\leq -D_U|U|_{H^1}^2 + k_1{}_{H^{-1}}\langle f(\lvert U\rvert_{L^2}, U),U\rangle_{H^1} \\
	    &\quad + k_2|U|_{L^2}|V|_{L^2} \\
		&\lesssim -D_U|U|_{H^1}^2 + k_1\langle \partial_uf(\lvert U\rvert_{L^2}, \widetilde V)U,U\rangle_{L^2} \\
		&\quad + C\lVert(U,V)\rVert_{\mathcal{H}}^2 \\
		&\lesssim -D_U|U|_{H^1}^2 + C|U|_{L^2}^2 + C\lVert(U,V)\rVert_{\mathcal{H}}^2
\end{align*}
for some $\widetilde V:\mathcal{D}\rightarrow\R$, again using \eqref{eq:fBoundedAbove}. 
Finally,
\begin{align*}
	\lvert A_1(U)\rvert_{H^{-1}}^2 &\lesssim D_U|U|_{H^1}^2 + k_1|f(\lvert U\rvert_{L^2}, U)|_{H^{-1}}^2 + k_2|V|_{L^2}^2,
\end{align*}
so it remains to control $|f(\lvert U\rvert_{L^2},U)|_{H^{-1}}^2\lesssim |f(\lvert U\rvert_{L^2},U)|_{L^1}^2\lesssim |U|_{L^1}^2+|U^2|_{L^1}^2+|U^3|_{L^1}^2$, and we have $|U|_{L^1}^2\lesssim|U|_{L^2}^2$ as well as $|U^2|_{L^1}^2=|U|_{L^2}^4$ and $|U^3|_{L^1}^2 = |U|_{L^3}^6\lesssim |U|_{H^{1/3}}^6\lesssim |U|_{H^1}^2|U|_{L^2}^4$. Thus $(H4)$ is true. Putting things together, we get that \cite{LiuRockner15} is applicable for sufficiently large $p\geq 1$, and the claim follows. 
\end{proof}

In order to improve \eqref{eq:FNBasicRegularity} to the $H^1$-norm, it suffices to prove coercivity, i.e. $(H3)$ from \cite{LiuRockner15}, for the Gelfand triple $L^2\subset H^1\subset H^2$. 
\begin{lemma}\label{lem:FNRegularityH1}
	The solution $(U,V)$ to \eqref{eq:Activator}, \eqref{eq:Inhibitor} satisfies $(A_1)$, i.e. for any $p\geq 1$:
	\begin{align}\label{eq:FNRegularityH1}
		\mathbb{E}\left[\sup_{0\leq t\leq T}|U_t|_{H^1}^p\right] < \infty, \quad \mathbb{E}\left[\sup_{0\leq t\leq T}|V_t|_{H^1}^p\right] < \infty.
	\end{align}
\end{lemma}
\begin{proof}
	This is done by
	\begin{align*}
		{}_{L^2}\langle A_1(U,V),U\rangle_{H^2} &\leq -D_U|U|_{H^2}^2 + k_1\langle \partial_uf(\lvert U\rvert_{L^2}, U)\nabla U,\nabla U\rangle_{L^2} \\
		    &\quad + k_2|U|_{H^1}|V|_{H^1} \\
			&\lesssim -D_U|U|_{H^2}^2 + C|U|_{H^1}^2 + C\lVert(U,V)\rVert_{H^1\times H^1}^2,
	\end{align*}
	where \eqref{eq:fBoundedAbove} has been used componentwise. By \cite[Lemma 5.1.5]{LiuRockner15} we immediately obtain \eqref{eq:FNRegularityH1}. 
\end{proof}

Remember that by integrating \eqref{eq:Inhibitor}, we can write \eqref{eq:Activator} as
\begin{align}
	\mathrm{d}U_t = (D_U\Delta U_t + \theta_1F_1(U_t) + \theta_2F_2(U_t) + \theta_3F_3(U)(t))\mathrm{d}t + B\mathrm{d}W_t,
\end{align}
where we adopted the notation from \eqref{eq:FNmodelF1}, \eqref{eq:FNmodelF2}, \eqref{eq:FNmodelF3}.

\begin{lemma}
	The process $(U, V)$ satisfies $(A_s)$ for some $s>1$. 
\end{lemma}
\begin{proof}
	Let $U=\overline U+\widetilde U$ be the decomposition of $U$ into its linear and nonlinear part as in \eqref{eq:SplittingEquationLinear},\eqref{eq:SplittingEquationNonlinear}. We will prove
	\begin{align}\label{eq:LpRegularity}
		\mathbb{E}\left[\sup_{0\leq t\leq T}|\widetilde U_t|_{W^{s, q}}^p\right] < \infty
	\end{align}
	for any $p,q\geq 1$ and $s<2$. As $\gamma>d/4$ in $d\leq 2$, there is $s>1$ such that \eqref{eq:LpRegularity} is true for $\overline U$ with $q=2$, and the claim follows. Let us now prove \eqref{eq:LpRegularity}. 
	First, note that by Lemma \ref{lem:FNRegularityH1} and the Sobolev embedding theorem in dimension $d\leq 2$, \eqref{eq:LpRegularity} is true for $U_t$ with any $p,q\geq 1$ and $s=0$. 
	For $k\in\N$, by $|U_t^k|_{L^q}^p = |U_t|_{L^{kq}}^{kp}$, we see that polynomials in $U$ satisfy \eqref{eq:LpRegularity} with $s=0$, too. In particular, using that $a$ is bounded,
	\begin{align}
		\mathbb{E}\left[\sup_{0\leq t\leq T}|F_i(U_t)|_{L^q}^p\right] < \infty
	\end{align}
	for $i=1,2$ and $p,q\geq 1$. A simple calculation 
	shows that the same is true for $F_3$. As in the proof of Proposition \ref{prop:AdditionalRegularity}, we see that
	\begin{align*}
		\sup_{0\leq t\leq T}|\widetilde U_t|_{\eta, q} \leq |U_0|_{\eta, q} + \frac{2}{\epsilon}T^\frac{\epsilon}{2}\sup_{0\leq t\leq T}|\theta_1F_1(U)+\theta_2F_2(U)+\theta_3F_3(U)|_{L^{q}}
	\end{align*}
	with $\eta<2$, $q\geq 1$ and $\epsilon=2-\eta$, and the proof is complete. 
\end{proof}

\begin{proof}[Proof of Proposition \ref{prop:FNWellPosed}]
	In $d\leq 2$, $H^s$ is an algebra for $s>1$, i.e. $uv\in H^s$ for $u,v\in H^s$ \cite{AdamsFournier03}. Together with the assumption that $a$ is bounded, it follows immediately that $(F_{s, \eta})$ holds for $F_1$, $F_2$ separately (with $K=0$) for any $s>1$ and $\eta < 2$. For $F_3$, we have for any $\delta>0$:
	\begin{align}
		\sup_{0\leq t\leq T}|F_3(U_t)|_{s+2-\delta}
			&\leq \int_0^T|e^{(t-r)(D_V\Delta-\epsilon I)}U_r|_{s+2-\delta}\mathrm{d}r \\
			&\leq \int_0^T(t-r)^{-1+\delta/2}|U_r|_{s}\mathrm{d}r \\
			&\leq \frac{2}{\delta}T^\frac{\delta}{2}\sup_{0\leq t\leq T}|U_t|_s,
	\end{align}
	so $F_3$ satisfies $(F_{s, \eta})$ even with $s\in\R$, $\eta<4$. It is now clear that the $(F_{s, \eta})$ holds for $F = \theta_1F_1+\theta_2F_2+\theta_3F_3$ for any $s>1$ and $\eta < 2$ (with $K=3$). 
	The claim follows from Proposition \ref{prop:AdditionalRegularity} (i) for $U$ and (ii) for $V$, noting that $V = \epsilon b F_3(U)$. 
\end{proof}

\subsection{Proof of Proposition \ref{prop:LinearModelEstimatorBounded}}\label{app:LinearModelEstimatorBounded}

This section is devoted to proving Proposition \ref{prop:LinearModelEstimatorBounded}. 
A similar argument can be found in \cite[Chapter 3]{Huebner93} for linear SPDEs. 
We start with an auxiliary statement, which characterizes the rate of $\det(A_N(X))$. 
For simplicity of notation, we abbreviate $a^N_i:=A_N(X)_{i,i}$ for $i=0,\dots, K$. 
While there is a trival upper bound 
$\det(A_N(X))\lesssim a^N_0\dots a^N_K$ 
obtained by the Cauchy-Schwarz inequality, the corresponding lower bound requires more work. Our argument is geometric in nature: A Gramian matrix measures the (squared) volume of a parallelepiped spanned by the vectors defining the matrix. The argument takes place in the separable Hilbert space $L^2(0,T;H^{2\alpha})$, and we find a (uniform in $N$) lower bound on the $(K+1)$-dimensional volume of the parallelepiped spanned by the vectors $P_N(-\Delta)X, P_NF_1(X), \dots, P_NF_K(X)$. While the latter $K$ components (and thus their $K$-dimensional volume) converge, the first component gets eventually orthogonal to the others as $N$ grows. This is formalized as follows: 

\begin{lemma}\label{lem:DeterminantRate}
	Let $\eta,s_0>0$ such that $(A_s)$ and $(F_{s,\eta})$ are true for $s_0\leq s < 2\gamma+1-d/2$. 
	Let $\gamma-d/4-1/2<\alpha\leq\gamma-d/4-1/2+\eta/2$. Under assumption $(L_\alpha)$, we have
	\begin{align}
		\det(A_N(X))\gtrsim a^N_0\dots a^N_K.
	\end{align}
	In particular, $A_N(X)$ is invertible if $N$ is large enough. 
\end{lemma}

\begin{proof}
	The argument is pathwise, so we fix a realization of $X$. We abbreviate $\langle\cdot,\cdot\rangle_\alpha:=\langle\cdot,\cdot\rangle_{L^2(0,T;H^{2\alpha})}$ and $\lVert\cdot\rVert_\alpha:=|\cdot|_{L^2(0,T;H^{2\alpha})}$, further we write $F_0(X) = \Delta X$ in order to unify notation. 
	By a simple normalization procedure, it suffices to show that 
	\begin{align}\label{eq:DeterminantFull}
		\liminf_{N\rightarrow\infty}\det\left(\left(\left\langle \frac{P_NF_i(X)}{\lVert P_NF_i(X)\rVert_\alpha},\frac{P_NF_j(X)}{\lVert P_NF_j(X)\rVert_\alpha}\right\rangle_\alpha\right)_{i,j=0,\dots K}\right) > 0.
	\end{align}
	Let $\epsilon>0$ and choose $N_0,N_1\in\N$ with $N_0\leq N_1$ such that:
	\begin{itemize}
		\item $\lVert P_NF_i(X)-F_i(X)\rVert_\alpha <\epsilon \lVert F_i(X)\rVert_\alpha$ and $\lVert F_i(X)\rVert_\alpha / \lVert P_NF_i(X)\rVert_\alpha < 2$ for $i=1,\dots,K$ and $N\geq N_0$. This is possible due to $\alpha\leq\gamma-d/4-1/2+\eta/2.$ 
		\item $\lVert P_{N_0}F_0(X) \,/\, \lVert P_{N}F_0(X)\rVert_\alpha\,\rVert_\alpha <\epsilon$ for $N\geq N_1$. This is possible as $\lVert P_NF_0(X)\rVert_\alpha$ diverges due to $\alpha > \gamma-d/4-1/2$. 
	\end{itemize}
	Now one performs a Laplace expansion of \eqref{eq:DeterminantFull} in the first column. Each entry of the matrix in \eqref{eq:DeterminantFull} can be bounded by $1$ trivially, but we can say more:
	\begin{align*}
		\left|\left\langle\frac{P_NF_0(X)}{\lVert P_NF_0(X)\rVert_\alpha},\frac{P_NF_i(X)}{\lVert P_NF_i(X)\rVert_\alpha}\right\rangle_\alpha\right| 
			&\leq \left|\left\langle\frac{P_{N_0}F_0(X)}{\lVert P_{N}F_0(X)\rVert_\alpha},\frac{P_{N_0}F_i(X)}{\lVert P_NF_i(X)\rVert_\alpha}\right\rangle_\alpha\right| \\
			&\hspace{-2cm}\quad + \left|\left\langle\frac{(P_N-P_{N_0})F_0(X)}{\lVert P_{N}F_0(X)\rVert_\alpha},\frac{(P_N-P_{N_0})F_i(X)}{\lVert P_NF_i(X)\rVert_\alpha}\right\rangle_\alpha\right| \\
			&\hspace{-2cm}\leq \left\lVert\frac{P_{N_0}F_0(X)}{\lVert P_{N}F_0(X)\rVert_\alpha}\right\rVert_\alpha + \left\lVert\frac{(I-P_{N_0})F_i(X)}{\lVert P_NF_i(X)\rVert_\alpha}\right\rVert_\alpha \\
			&\hspace{-2cm}\leq \epsilon + 2\epsilon
	\end{align*}
	for $N\geq N_1$ with $i=1,\dots,K$, where we used the Cauchy-Schwarz inequality and the choice of $N_0$, $N_1$. Consequently, in order for \eqref{eq:DeterminantFull} to be true, it suffices to show for the $(0, 0)$-minor:
	\begin{align}\label{eq:DeterminantMinor}
		\liminf_{N\rightarrow\infty}\det\left(\left(\left\langle \frac{P_NF_i(X)}{\lVert P_NF_i(X)\rVert_\alpha},\frac{P_NF_j(X)}{\lVert P_NF_j(X)\rVert_\alpha}\right\rangle\right)_{i,j=1,\dots K}\right) > 0.
	\end{align}
	By $(L_\alpha)$, we have
	\begin{align}\label{eq:DeterminantLimit}
		\det\left(\left(\left\langle \frac{F_i(X)}{\lVert F_i(X)\rVert_\alpha},\frac{F_j(X)}{\lVert F_j(X)\rVert_\alpha}\right\rangle_\alpha\right)_{i,j=1,\dots K}\right) > 0,
	\end{align}
	so \eqref{eq:DeterminantMinor} holds true by continuity 
	for $i=1,\dots,K$.
\end{proof}

\begin{proof}[Proof of Proposition \ref{prop:LinearModelEstimatorBounded}]
	As before, we formally write $F_0(X) = \Delta X$ in order to unify notation. 
	By plugging in the dynamics of $X$, we see that for $i=0,\dots,K$,
	\begin{align}
		b_N(X)_{i} = \int_0^T\langle(-\Delta)^{2\alpha-\gamma}P_NF_i(X),\mathrm{d}W^N_t\rangle + \sum_{k=0}^K\theta_kA_N(X)_{i, k}.
	\end{align}
	Thus, with $\bar b_N(X)_{i} = \int_0^T\langle(-\Delta)^{2\alpha-\gamma}P_NF_i(X),\mathrm{d}W^N_t\rangle$ and $\theta = (\theta_0,\dots,\theta_K)^T$, this read as
	\begin{align}
		A_N(X)\hat\theta^N(X) = \bar b_N(X) + A_N(X)\theta,
	\end{align}
	i.e.
	\begin{align}
		\hat\theta^N - \theta = A_N^{-1}(X)\bar b_N(X).
	\end{align}
	Using the explicit representation for the inverse matrix $A_N(X)^{-1}$, we have 
	\begin{align}
		\hat\theta^N_j-\theta_j = \frac{1}{\det(A_N(X))}\sum_{i=0}^K(-1)^{i+j}\det(\overline A^{(i,j)}_N(X))\bar b_N(X)_i,
	\end{align}
	where $\overline A^{(i,j)}_N(X)$ results from $A_N(X)$ by erasing the $i$th row and the $j$th column. 
	Remember that by Lemma \ref{lem:DenominatorRate}, $a^N_0\asymp C_\alpha N^{\frac{2}{d}(2\alpha-2\gamma+1)+1}$. 
	Further, due to $\alpha\leq \gamma-d/4-1/2+\eta/2$, $a^N_i / a^N_0 \rightarrow 0$ in probability for all $i=1,\dots,K$. 
	Now, by the Cauchy-Schwarz inequality, each summand in the numerator can be bounded by terms of the form
	\begin{align}
		\prod_{k=0}^Ka^N_k\frac{|\bar b_N(X)_i|}{\sqrt{a^N_ia^N_j}}
	\end{align}
	for $i,j=0,\dots,K$. Thus, by Lemma \ref{lem:DeterminantRate}, in order to prove the claim it remains to find a uniform bound in probability for the terms of the form $|\bar b_N(X)_i| / \sqrt{a^N_ia^N_j}$. Now, for $i=1,\dots,K$,
	\begin{align}
		\mathbb{E}{\left[\bar b_N(X)_{i}^2\right]} &= \mathbb{E}\left[\left(\int_0^T\langle(-\Delta)^{2\alpha-\gamma}P_NF_i(X),\mathrm{d}W^N_t\rangle\right)^2\right] \\
		&= \mathbb{E}\left[\int_0^T|(-\Delta)^{2\alpha-\gamma}P_NF_i(X)|_H^2\mathrm{d}t\right] < \infty
	\end{align}
	uniformly in $N$ as $\alpha < \gamma - d/8 - 1/4 + \eta/4$. Further, $(a^N_0)^{-1/2}$ converges to zero in probability, while $(a^N_i)^{-1/2}$ converges to a positive constant for $i=1,\dots,K$. 
	For $i=0$, the reasoning is similar: By Propositon \ref{prop:AdditionalRegularity} and Lemma \ref{lem:DenominatorRate}, we have
	\begin{align}
	    \mathbb{E}\left[\bar b_N(X)_0^2\right] \asymp \mathbb{E}\left[\int_0^T|(-\Delta)^{1+2\alpha-\gamma}\overline X^N_t|_H^2\mathrm{d}t\right] \asymp C_{2\alpha-\gamma}N^{\frac{2}{d}(4\alpha-4\gamma+1)+1}.
	\end{align}
	Together with $a_0^N\sim N^{\frac{2}{d}(2\alpha-2\gamma+1)+1}$ and $\alpha\leq\gamma$, this yields that $\bar b_N(X)_0^2/a^N_0$ (and consequently $|\bar b_N(X)_0|/\sqrt{a^N_0a^N_j}$ for $j=0,\dots,K$) is bounded in probability.
	The claim follows. 
\end{proof}

\section{Methods}\label{app:methods}

\subsection*{Numerical simulation}
For the evaluation in Section \ref{sec:data_analysis} 
and Figure \ref{fig:Experiment_and_Simulations} (F) 
we simulated \eqref{eq:Activator}, \eqref{eq:Inhibitor} on a square with side $L=75$, with spatial increment $\mathrm{d}x=0.375$, for $t\in [T_0, T_1)$, $T_0=500$, $T_1=700$, 
with temporal increment $\mathrm{d}t=0.01$, using an explicit finite difference scheme, of which we recorded every $100$th frame. We choose $T_0=500$ in order to avoid artefacts stemming from zero initial conditions at $T=0$.
We set $b=0.2$, 
$\gamma=0$, $\sigma=0.1$ 
and $a(x) = 0.5-b+0.5(x/(0.33u_0L^2)-1)$. 
In Section \ref{sec:SimulationSigma}, we use additional simulations with $\sigma=0.05$ and $\sigma=0.2$.

For single and giant cell simulations 
in Figure \ref{fig:Experiment_and_Simulations} (D), (E) 
we solve \eqref{eq:Activator}, \eqref{eq:Inhibitor} 
on a square with side $L=75$, with spatial increment $\mathrm{d}x=0.15$, for $t\in[0,T]$, $T=1000$, with temporal increment $\mathrm{d}t=0.002$, using an explicit finite difference scheme together with a phase field model \cite{FlemmingFontAlonsoBeta20} to account for the interior of the single and giant cells, corresponding, respectively, to an area of A$_0$=113 \si{\micro}$m^2$ and  A$_0$=2290 \si{\micro}$m^2$. 
We used 
$b=0.4$ 
and $a(x) = 0.5-b+0.5(x/(0.25u_0A_0)-1)$. 
The noise terms are chosen as in \cite{FlemmingFontAlonsoBeta20}, with $\sigma=0.1$ and $k_\eta=0.1$ in the notation therein.

For both simulations, the remaining parameters are 
$D_U=0.1, D_V=0.02, k_1=k_2=1, u_0=1, \epsilon=0.02$. 
The unit length is 1~\si{\micro}m, the unit time is $1s$.

\subsection*{Experimental settings}
For experiments \textit{D. discoideum} AX2 cells expressing mRFP-LimE$\Delta$ as a marker for F-actin were used. Cell culture and electric pulse induced fusion to create giant cells were performed as described in \cite{Gerhardt_actin_2014,FlemmingFontAlonsoBeta20}. For live cell imaging an LSM 780 (Zeiss, Jena) with a 63x objective lens (alpha Plan-Apochromat, NA 1.46, Oil Korr M27, Zeiss Germany) were used. The spatial and temporal resolution was adjusted for each individual experiment to acquire image series with the best possible resolution, while protecting the cells from photo damage. Images series were saved as 8-bit tiff files. For image series where the full range of 256 pixel values was not utilized, the image histograms were optimized in such a way that the brightest pixels corresponded to a pixel value of 255.

\section*{Acknowledgements} 
This research has been partially funded by Deutsche Forschungsgemeinschaft (DFG) - SFB1294/1 - 318763901. 
SA acknowledges funding from MCIU and FEDER -- PGC2018-095456-B-I00.

\bibliographystyle{amsalpha}
\bibliography{Literatur,refs_CB}

\end{document}